\documentclass[letterpaper,11pt]{article}

\usepackage{amssymb,a4wide}
\usepackage{times,upgreek,bbm}

\usepackage{xspace}

\usepackage{graphicx}

\newcommand{\Xomit}[1]{ }

\newtheorem{theorem}{Theorem}

\newtheorem{corollary}[theorem]{Corollary}

\newtheorem{lemma}[theorem]{Lemma}

\newenvironment{proof}[1][Proof]{\textbf{#1.} }{\ \rule{0.5em}{0.5em}}
\newcommand{\eps}{\varepsilon}

\newcommand{\Y}{\mathbf {\cal Y}}
\newcommand{\X}{\mathbf {\cal X}}

\newcommand{\B}{{\cal{B}}}
\newcommand{\A}{{\cal{A}}}
\newcommand{\R}{{\cal{R}}}
\newcommand{\W}{{\cal{W}}}

\newcommand{\C}{{\cal{C}}}

\newcommand{\Prob}{{\sc{GCBP'$_\eps$}}}

\newcommand{\gc}{{\sc{GCBP}}}
\newcommand{\opt}{\mbox{\textsc{opt}}}

\newcommand{\ffo}{\mbox{\textsc{ff}}}
\newcommand{\bfo}{\mbox{\textsc{bf}}}
\newcommand{\nfo}{\mbox{\textsc{nf}}}
\newcommand{\ffd}{\mbox{\textsc{ffd}}}
\newcommand{\bfd}{\mbox{\textsc{bfd}}}
\newcommand{\nfd}{\mbox{\textsc{nfd}}}
\newcommand{\ffi}{\mbox{\textsc{ffi}}}
\newcommand{\bfi}{\mbox{\textsc{bfi}}}
\newcommand{\nfi}{\mbox{\textsc{nfi}}}
\newcommand{\fnfi}{\mbox{\textsc{fnfi}}}

\newcommand{\mh}{\mbox{\textsc{mh}}}

\begin{document}

\title{Bin packing with general cost structures}

\author{Leah Epstein\thanks{Department of Mathematics, University of Haifa, 31905 Haifa, Israel.
{\tt lea@math.haifa.ac.il}.} \and Asaf Levin
\thanks{Chaya fellow. Faculty of Industrial Engineering and
Management, The Technion, Haifa, Israel. {\tt
levinas@ie.technion.ac.il.}}}
\date{}

\maketitle

\vspace{-0.3cm}
\begin{abstract}
Following the work of Anily et al., we consider a variant of bin
packing, called {\sc bin packing with general cost structures}
(\gc) and design an asymptotic fully polynomial time
approximation scheme (AFPTAS) for this problem. In the classic bin
packing problem, a set of one-dimensional items is to be assigned
to subsets of total size at most 1, that is, to be packed into
unit sized bins. However, in \gc, the cost of a bin is not 1 as
in classic bin packing, but it is a non-decreasing and concave
function of the number of items packed in it, where the cost of
an empty bin is zero. The construction of the AFPTAS requires
novel techniques for dealing with small items, which are
developed in this work. In addition, we develop a fast
approximation algorithm which acts identically for all
non-decreasing and concave functions, and has an asymptotic
approximation ratio of 1.5 for all functions simultaneously.
\end{abstract}

\section{Introduction}
\label{sec:intro} Classic bin packing
\cite{Ullman71,CsiWoe98,CoGaJo97,Gon32} is a well studied problem
which has numerous applications. In the basic variant of this
problem, we are given $n$ items of size in $(0,1]$ which need to
be assigned to unit size bins. Each bin may contain items of
total size at most 1, and the goal is to minimize the number of
bins used.

Consider the following possible application. A multiprocessor
system, where each bin represents one processor, is available for
one unit of time. However, a processor that executes a large
number of short tasks causes the system a larger load than a
processor that executes a smaller number of long tasks, even if
the total duration of the tasks is equal in both cases. This is
one motivation to the problem {\sc bin packing problem with
general cost structures} (\gc) that we study here. The problem has
additional applications in reliability, quality control and
cryptography \cite{ABSL}.

In the problem \gc, the cost of a bin is not a unit cost, but
depends on the number of items actually packed into this bin. More
precisely, we define the problem as follows.  The input consists
of $n$ items $I=\{ 1,2,\ldots ,n\}$ with sizes $1\geq s_1\geq
s_2\geq \cdots \geq s_n\geq 0$, and a function $f:\{ 0,1,2,\ldots
,n\} \rightarrow \mathbb{R}^+_0$, where $f$ is a monotonically
non-decreasing concave function, for which $f(0)=0$. The goal is
to partition $I$ into some number of sets $S_1,\ldots ,S_m$,
called bins, such that $\sum_{j\in S_i} s_j \leq 1$ for any $1\leq
i \leq m$, and so that $\sum_{i=1}^m f(|S_i|)$ is minimized. We
say that a function $f$ is {\it valid} if it has the properties
above, and an instance of \gc\ is defined not only by its input
item sizes but also using the function $f$. We assume that
$f(1)=1$ (otherwise we can apply scaling to the cost function
$f$).

Anily, Bramel and Simchi-Levi~\cite{ABSL} introduced \gc\ and
described the applications in detail. We describe their results in
what follows.  Further results on \gc\ appear in \cite{BRS98}, but
these additional results are not related to this paper. A related
model was studied by Li and Chen \cite{LC06}. In this model the
cost of a bin is a concave and monotonically non-decreasing
function of the {\it total size} of items in it.

 For an algorithm ${\mathcal{A}}$, we
denote its cost by ${\mathcal{A}}$ as well. The cost of an optimal
algorithm is denoted by \opt. We define the asymptotic
approximation ratio of an algorithm ${\mathcal{A}}$ as the infimum
${\mathcal{R}}\geq 1$ such that there exists a constant $c$, which
is independent of the input, so that any input satisfies $
{\mathcal{A}} \leq{\mathcal{R}} \cdot \mbox{\textsc{opt}} +c$. The
absolute approximation ratio of an algorithm ${\mathcal{A}}$ is
the infimum ${\mathcal{R}}\geq 1$ such that for any input, $
{\mathcal{A}} \leq{\mathcal{R}} \cdot \mbox{\textsc{opt}}$. An
asymptotic polynomial time approximation scheme is a family of
approximation algorithms such that for every $\eps>0$ the family
contains a polynomial time algorithm with an asymptotic
approximation ratio of $1+\eps$. We abbreviate {\it asymptotic
polynomial time approximation scheme} by APTAS (also called an
asymptotic PTAS). An asymptotic fully polynomial time
approximation scheme (AFPTAS) is an APTAS whose time complexity is
polynomial not only in the input size but also in $1\over \eps$.
Polynomial time approximation schemes and  fully polynomial time
approximation schemes, which are abbreviated as PTAS and FPTAS,
are defined similarly, but are required to give an approximation
ratio of $1+\eps$, according to the absolute approximation ratio.

Anily, Bramel and Simchi-Levi \cite{ABSL} analyzed the worst case
performance of some natural bin-packing heuritics when they are
applied for \gc. They showed that many common heuristics for bin
packing, such as First Fit (\ffo), Best Fit (\bfo) and Next Fit
(\nfo), do not have a finite asymptotic approximation ratio. Even
an application of the first two heuristics on lists of items that
are sorted by size in a non-increasing order, i.e., the algorithms
First Fit Decreasing (\ffd) and Best Fit Decreasing (\bfd),  leads
to similar results. However, Next Fit Decreasing (\nfd) behaves
differently, and was shown to have an asymptotic approximation
ratio of exactly 2. Sorting the items in the opposite order gives
a better asymptotic approximation ratio of approximately 1.691 (in
this case, the three algorithms First Fit Increasing (\ffi), Best
Fit Increasing (\bfi) and Next Fit Increasing (\nfi) are the same
algorithm). Note that these heuristics are independent of the
specific function $f$. It is stated in \cite{ABSL} that any
heuristic that is independent of $f$ has an asymptotic
approximation ratio of at least $\frac 43$. Therefore, finding an
algorithm with a smaller asymptotic approximation ratio, and
specifically, an asymptotic approximation scheme, requires a
strong usage of the specific function $f$.

In this paper, we develop an AFPTAS for \gc. We develop a
framework, where the action of the scheme for a given
non-decreasing concave function $f$ with $f(0)=0$ is based on its
exact definition. We also develop a new approximation algorithm
{\sc{MatchHalf}} (\mh), which acts obliviously of $f$, similarly
to the behavior of the algorithms of \cite{ABSL}. We prove that
our algorithm has an asymptotic approximation ratio of at most 1.5
for any non-decreasing concave function $f$ with $f(0)=0$,
improving over the tight bound of approximately 1.691, proved by
Anily et al. \cite{ABSL}, on the asymptotic approximation ratio of
\nfi.

The classic bin packing problem is clearly a special case of \gc\
as one can set $f(0)=0$ and $f(i)=1$ for all $i\geq 1$, where the
resulting function is monotonically non-decreasing and concave.
Therefore, \gc\ inherits the hardness proof of the classic bin
packing problem.  That is, \gc\ cannot be approximated within an
absolute factor better than $3\over 2$ (unless $P=NP$).  This
motivates our use of asymptotic approximation ratio as the main
analytic tool to study approximation algorithms for \gc.  In this
metric we design the best possible result (assuming $P\neq NP$),
i.e., an AFPTAS.

A study of this nature, where approximation schemes are developed
for bin packing type problems, and in particular, where the
complexity of such a problem is completely resolved by designing
an AFPTAS, is an established direction of research. Studies of
similar flavor were widely conducted for other variants of bin
packing, see e.g. \cite{KK82,JS05,Mur87,SY07,EL07afptas}.

Fernandez de la Vega and Lueker \cite{FerLue81} showed that the
classic bin packing problem admits an APTAS. This seminal work
introduced rounding methods which are suitable for bin packing
problems. These methods, which were novel at that time, are widely
used nowadays. Karmarkar and Karp \cite{KK82} employed these
methods together with column generation and designed an AFPTAS
\cite{KK82}. In \cite{EL07afptas}, the complexity of two variants
of bin packing with unit sized bins are resolved, that is, an
AFPTAS is designed for each one of them. The first one is {\it Bin
packing with cardinality constraints} \cite{KSS75,CKP03}, in which
an additional constraint on the contents of a bin is introduced.
Specifically, there is a parameter $k$ which is an upper bound on
the number of items that can be packed in one bin. The goal is as
in classic bin packing, to minimize the number of bins used. The
second one is {\it Bin packing with rejection}
\cite{Eps06,BCH08,DoH05}, in which each item has a rejection
penalty associated with it (in addition to the size). Each item
has to be either packed or rejected, and the goal is to minimize
the sum of the following two factors: the number of bins used for
the packed items and the total rejection cost of all rejected
items. Note that prior to the work of \cite{EL07afptas}, these two
problems were already known to admit an APTAS
\cite{CKP03,Eps06,BCH08}. The main new tool, used in
\cite{EL07afptas}, which allows the design of schemes whose
running time is polynomial in $\frac 1\eps$, is a treatment for
small items using new methods developed in that work. The
treatment of small enough items for the classic problem is rather
simple. Roughly, the small items can be put aside while finding a
good approximate solution, and can be added later in any
reasonable fashion. Already in \cite{CKP03}, it was shown that if
the same treatment is applied to small items in the case of
cardinality constraints, this leads to poor approximation ratios.
Therefore, Caprara, Kellerer and Pferschy \cite{CKP03} developed
an alternative method for dealing with small items. This method
still separates the packing of large items from the packing of
small items. The scheme enumerates a large number of potential
packings of the large items, and for each packing, tests the
quality of a solution that is constructed by adding the small
items to the packing in a close to optimal way. The enumeration
prevents this method from being used for designing algorithms with
running time which is polynomial in $\frac 1\eps$. The way to
overcome this difficulty, used in \cite{EL07afptas}, is to find a
good packing of large items, that takes into account the existence
of small items, and allocates space for them. The packing of large
items is typically determined by a linear program, therefore, the
linear program needs to define at least some properties for the
packing of small items. Specifically, the linear program does not
decide on the exact packing of small items, but only on the type
of a bin that they should join, where a type of a bin is defined
according to the size of large items in the bin for bin packing
with rejection, and on both the size and number of large items,
for bin packing with cardinality constraints.

The problem studied in this paper, \gc, is more complex than the
ones of \cite{EL07afptas} in the sense that the cost of a bin is
not just 1. Therefore, even though cardinality constraints are not
present, the number of items packed into each bin must be
controlled, in order to be able to keep track of the cost of this
bin. In classic bin packing, and other well known variants,
forcing all the bins of a solution to be completely occupied,
results in a perfect solution. To demonstrate the difficulty of
\gc, we show the existence of a non-decreasing concave function
$f$ with $f(0)=0$, for which such a solution may still lead to a
poor performance with respect to $f$.

In our scheme, cardinality constraints are implied by an advanced
decision on the cost that needs to be paid for a given bin, that
becomes a part of the type of the bin. The specific packing of
small items, which is based on the output of the linear program,
needs to be done carefully, so that the solution remains feasible,
and to avoid large increases in the cost of the solution. An
additional new ingredient used in our AFPTAS is a pre-processing
step, which is performed on small items, where some of them are
packed in separate bins which are not used for any other items. In
typical packing problems, bins which contain only very small items
are relatively full, and thus the additional cost from such bins
is close to the total size of these items. However, in our case,
such a bin usually contains many items, and may result in a high
cost. Therefore, our scheme always packs some portion of the
smallest items separately, before any methods of packing items
through a linear program are invoked. We show that the increase in
the cost of the solution, due to the pre-processing step, is small
enough, yet this allows more flexibility in the treatment of other
small items, i.e., an additional bin would have a small cost
compared to \opt.

The structure of the paper is as follows. In Section \ref{prem} we
supply examples showing the unique nature of the problem \gc,
accompanied with new properties and some properties used in
previous work. We use all these properties later in the paper. We
introduce our fast approximation algorithm and analyze it in
Section \ref{better}. Our main result is given in Section
\ref{afptas}.

\section{Preliminaries}
\label{prem} In this section we demonstrate the differences
between classic bin packing problems, and \gc. We also state some
properties proved in \cite{ABSL} and \cite{BC81} to be used later.

As mentioned in the introduction, common heuristics do not have a
finite approximation ratio for \gc\ \cite{ABSL}, and other
heuristics have a higher approximation ratio than one would
expect. Another difference is that sorting items in a
non-decreasing order of their sizes is better than a
non-increasing order.

A class of (concave and monotonically non-decreasing) functions
$\{f_q\}_{q \in \mathbb{N}}$ that was considered in \cite{ABSL}
is the following. These are functions that grow linearly (with a
slope of 1) up to an integer point $q$, and are constant starting
from that point. Specifically, $f_q(t)=t$ for $t\leq q$ and
$f_q(t)=q$ for $t > q$. It was shown in \cite{ABSL} that focusing
on such functions is sufficient when computing upper bounds on
algorithms that act independently of the function.

For an integer $K>2$, consider inputs consisting of items of two
sizes; $a=1-\frac 1K$, and $b=\frac 1{K^2}$.

Assume first that there is a single item of size $a$, and $2K$
items of size $b$. \nfd\ packs the large item together with $K$ of
the small items in one bin, and additional $K$ items in another
bin. Consider the function $f_K$. The cost of the solution is
$f_K(K+1)+f_K(K)=2K$. A solution that packs all small items in
one bin and the large item in another bin has a cost of
$f_K(1)+f_K(2K)=K+1$. Thus, even though both packings use the
same number of bins, the cost of the first packing, which is
produced by \nfd, is larger by a factor that can be made
arbitrarily close to 2, than the cost of the second packing.
Moreover, even though only two bins are used, this proves an {\it
asymptotic} lower bound of 2 on the approximation ratio of \nfd\
(this bound is tight due to \cite{ABSL}).

Assume now that there are $K$ items of size $a$ and $K^2$ items of
size $b$. An optimal packing for the classic bin packing problem
clearly consists of $K$ bins, such that each one is packed with
one large item and $K$ small items. Using the function $f_{K}$,
this gives a cost of $K^2$. A different packing collects all small
items in one bin, and has the cost $K\cdot f_K(1)+f_K(K^2)=2K$.
Since $K$ can be chosen to be arbitrarily large, we get that the
first packing, which is the unique optimal packing in terms of the
classic bin packing problem, does not have a finite approximation
ratio. Note that this first packing would be created by \ffd\ and
\bfd, and also by \ffo, \bfo\ and \nfo, if the input is sorted
appropriately.

Throughout the paper, if a specific cost function $f$ is
considered, we use \opt\ to denote the cost of an optimal solution
\opt\ for the original input, which is denoted by $I$, with
respect to $f$. For an input $J$ we use $\opt(J)$ to denote both
an optimal solution (with respect to $f$) for the input $J$ (where
$J$ is typically an adapted input), and its cost. Thus
$\opt=\opt(I)$. For a solution of an algorithm $\A$, we denote by
$m(\A)$ the number of bins in this solution. For an input $I$ we
let $\min(I)$ to be cost of an optimal solution with respect to
the function $f_k$ for $k=1$, that is, with respect to classic bin
packing. We let $f_k(\A(I))$ be the cost of an algorithm $\A$ on
$I$, calculated with respect to function $f_k$, and use $f_k(\A)$,
if $I$ is clear from the context.

We further state some lemmas proved in \cite{ABSL} that allow us
to simplify our analysis in the next section.
\begin{lemma}\label{NFDNFI} [Property 3 in \cite{ABSL}]
$f_1(\nfi(I))=f_1(\nfd(I))$, and therefore $\sum\limits_{i \in I}
w(s_i) \geq f_1(\nfi(I))-3$.
\end{lemma}
\begin{lemma}
\label{reduce} [Theorem 1 in \cite{ABSL}] Consider a packing
heuristic $\A$ that does not use information on the function $f$.
If the asymptotic approximation ratio of $\A$ is at most $\R$, for
any function $f_k$ (for $k\geq 1$), then the asymptotic
approximation ratio of $\A$ is at most $\R$ for any non-decreasing
concave function $f$ with $f(0)=0$.
\end{lemma}

A useful packing concept, defined in \cite{ABSL}, is {\it
consecutive bins}. Recall that we assume $s_1 \geq s_2 \geq \cdots
\geq s_n$. Let $B_1,B_2,\ldots,B_m$ be the subsets of items packed
into the bins created in some solution $\B$ that packs the items
in $m$ bins, where $B_i$ is the $i$-th bin. The packing has
consecutive bins if the union $\cup_{j \leq s} B_j$ is a suffix of
the sequence $1,2,\ldots,n$ for any $1 \leq s \leq m$. That is, if
the first $s$ bins contain $n'$ items, then these are the $n'$
items $n-n'+1,\ldots,n-1,n$ (and thus the smallest $n'$ items).
The following lemma states that \nfi\ is the ``best'' heuristic
among such with consecutive bins. Consider a given input $I$, the
cost function $f_k$ and a feasible packing with consecutive bins
$\B$.

\begin{lemma} \label{consec} [Corollary 3 in \cite{ABSL}]
 $\  \ f_k(\nfi(I))\leq
f_k(\B(I))$.
\end{lemma}

A partition of the items (which is not necessarily a valid
packing) with consecutive bins is called {\it an overflowed
packing} if for all $1<i<m$, $\sum\limits_{j \in B_i} s_j
>1$. Clearly, if $m>2$, such a packing must be infeasible. The
following lemma implies a lower bound on the cost of an optimal
solution. Consider a given input $I$, a cost function $f_k$, an
overflowed packing with consecutive bins $\B$, and a feasible
packing $\A$.

\begin{lemma} \label{overflow}[Corollary 1 in \cite{ABSL}]
$ \ \ f_k(\B(I))\leq f_k(\A(I))$.
\end{lemma}

Using these properties, in order to analyze \nfi, it is enough to
consider the functions $f_k$ for $k\geq 1$. It was shown in
\cite{ABSL} that the asymptotic approximation ratio of \nfi\ for
the function $f_k$ ($k\geq 2$) is at most $1+\frac 1k$. The
asymptotic approximation ratio of \nfi\ for $f_1$, that is, for
classic bin packing, follows from the results of \cite{BC81} and
from Lemma \ref{NFDNFI}. This ratio is
$\sum\limits_{i=1}^{\infty}\frac{1}{\pi_1-1}\approx 1.691$. Thus
the upper bound of $1.691$ \cite{ABSL} follows. In the next
section we use these properties to develop a new algorithm. The
algorithm needs to carefully keep the approximation ratio for
$k=2$ while improving the approximation ratio for $k=1$.


\section{A fast approximation algorithm \mh}
\label{better} In this section we describe a simple and fast
algorithm \mh, that does not need to know the function $f$ in
advance. This algorithm is a modification of \nfi\ that tries to
combine a part of the relatively large items (of size larger than
$\frac 12$) in bins together with one additional item. Note that
except for possibly one item, \nfi\ packs all such items in
dedicated bins.

As mentioned above, \nfi\ has an asymptotic approximation ratio of
at most $\frac{k+1}{k}$ for the function $f_k$ with $k\geq 2$.
Therefore, the difficult case is actually the classic problem. On
the other hand, using heuristics that perform well for the classic
problem, such as \ffd, may lead to worse results for $k\geq 2$
(which in fact is the case for \ffd). Therefore, we define an
algorithm that acts identically to \nfi, except for the usage of a
pre-processing step.

\begin{center}
\fbox{
\begin{minipage}{0.95\textwidth}
\noindent {\bf Algorithm {\sc MatchHalf} (\mh)}
\begin{enumerate}
\item Let $t$ be the number of items in $I$ with size in $(\frac
12,1]$ (which are called {\it large items}). \item \label{bb} Let
$M_0=\{\lceil \frac {t+1}2 \rceil, \ldots, t\}$, that is, $M_0$ is
the set of smallest $\lceil \frac {t}2 \rceil $ large items, and
let $M_1=\{1,\ldots,\lceil \frac {t-1}2 \rceil\}$ be the remaining
large items. Let $S=\{t+1,\ldots,n\}$ be called the set of small
items. \item Define the following bipartite graph. One set of
vertices consists of the large items of $M_0$. The other set of
vertices consists of all small items. An edge $(a,b)$ between
vertices of items of sizes $s_a>\frac 12$ and $s_b\leq \frac 12$
exists if $s_a+s_b\leq 1$, i.e., if these two items can be placed
in a bin together. If this edge occurs, its cost is defined as
$c(a,b)=w(b)$ (using the function $w$ of Section \ref{prem}).
\item\label{step4} Find a maximum cost matching in the bipartite graph. This
matching can actually be found using the following greedy process.
Insert the items of $S$ into a queue in a sorted order, with item
$t+1$ at the top, and the items $M_0$ are inserted into a queue in
a sorted order with item $t$ at the top. At each time, let $j$ be
the item at the top of the first queue, and $i$ the item at the
top of the second queue. If $s_i+s_j \leq 1$, these items are
matched, and removed from the queues. Otherwise, item $j$ cannot
be matched to any item of the second queue (since $s_i$ is minimal
in that queue), so $j$ is removed from the first queue. This
process is done until one of the queues is empty, and is performed
in linear time. \item Each pair of matched items is removed from
$I$. Every matched pair is packed into a bin together. \item{}
\label{sixth} Pack the remaining items using \nfi.
\end{enumerate}
\end{minipage}
}
\end{center}

The greedy process of step \ref{step4} finds an optimal matching
by a simple exchange argument. We note that only (approximately)
half of the large items are possibly matched in the pre-processing
step. A larger fraction may cause an asymptotic approximation
ratio above $1.5$, as can be seen in the following example. Let
$K$ be an integer such that $K>2$. The input set $I$ consists of
$K$ items of size $\frac 1{K}$ and $K$ items of size $1-\frac
1{K}$. Running \nfi\ on this input results in one bin containing
$K$ items of size $\frac 1{K}$ and $K$ bins containing one larger
item. However, if we match an $\alpha$ fraction (for some $0 \leq
\alpha \leq 1$) of the larger items in a pre-processing step,
there would be approximately $\alpha K$ bins with two items.
Consider the function $f_2$. We get $f_2(\nfi(I))=K+2$, whereas
the cost with pre-processing is at least $  \alpha K   +K$. This
would give an approximation ratio of at least $1+\alpha$.

For the analysis of \mh, we use  weighting functions. This type of
analysis  was widely used for classic bin packing, and many
variants of bin packing. The basic technique was used as early as
in 1971 by Ullman ~\cite{Ullman71} (see also
\cite{JoDUGG74,LeeLee85,Seiden02J}). We make use of adaptation of
the following function $w:[0,1]\rightarrow {\mathbb R}$ (that is
equal to the function $W_1(p)$ defined in \cite{BC81} for any
$p>0$). We first define the well known sequence $\pi_i$, $i\geq
1$, which often occurs in bin packing. Let $\pi_1 = 2$, and for
$i\geq 1$, $\pi_{i+1} = \pi_i (\pi_i -1) + 1$. Thus $\pi_2=3$,
$\pi_3=7$, $\pi_4=43$, etc. For $p\in (\frac 1{k+1},\frac 1k]$, we
define $w(p)=\frac 1k$, if $k=\pi_i-1$ for some $i \geq 1$, and
otherwise, $w(p)=\frac {k+1}k \cdot p$. Finally, we let $w(0)=0$.
Note that $w$ is a monotonically non-decreasing function. It was
shown in \cite{BC81} that for a given input $I$, $\sum\limits_{i
\in I} w(s_i) \geq f_1(\nfd(I))-3$. Even though both \cite{BC81}
and \cite{ABSL} assume that no zero sized items exist, clearly,
the number of bins used by \nfd\ and \nfi\ does not increase as a
result of the existence of such items, unless all input items are
of size zero, and therefore, this property on the weights still
holds even if zero sized items are allowed.

We start with proving the asymptotic approximation ratio for
$f_1$.

\begin{lemma}
\label{weights_mh} For any input $I$, $m(\mh(I))\leq \frac 32
\min(I) +3$.
\end{lemma}
\begin{proof}
We   use the following theorem.
\begin{theorem}
\label{weights} Consider an algorithm $\A$ for classic bin
packing. Let $w_1,w_2$ be two weight measures defined on the input
items, $w_i:I \rightarrow {\mathbb{R}}$, for $i=1,2$. Let $W_1(I)$
and $W_2(I)$ denote the sum of weights of all input items of $I$,
according to $w_1$ and $w_2$ respectively, and assume $W_2(I) \leq
W_1(I)$. Assume that for every input of the algorithm, the number
of bins used by the algorithm $\A$ is at most  $W_2(I)+\tau$, for
a constant value $\tau$ which is independent of $I$. Denote by
$W_I$ the supremum amount of weight that can be packed into a bin
of the optimal solution, according to measure $w_1$. Then the
asymptotic approximation ratio of $\A$ is no larger than $W_I$.
\end{theorem}
\begin{proof}
Given an input $I$ we have $\A \leq W_2(I)+\tau$. Since an optimal
algorithm has $\opt(I)$ bins, with a weight of at most $W_I$ in
each one of them, we get the upper bound on the weight, according
to $w_1$; $W_1(I)\leq W_I\cdot \opt(I)$. Using $W_2(I) \leq
W_1(I)$, we get $\A \leq W_I\opt(I)+\tau$ and the theorem follows.
\end{proof}

We define a weight measure $w_2$ on items as follows. For every
item $i$, we let $w_2(i)=w(s_i)$, except for small items that are
matched to large items in the pre-processing step of \mh. These
items receive a weight of zero according to $w_2$. Let $X$ be the
number of bins created by the pre-processing step and $Y$ the
number of bins created by \nfi\ (i.e., in Step \ref{sixth} of the
algorithm). Let $I'$ be the input after the removal of items in
the pre-processing step. By Lemma \ref{NFDNFI}, we have
$\sum\limits_{i\in I'}w_2(i)=\sum\limits_{i\in I'}w(s_i) \geq
Y-3$. On the other hand, every bin created in the pre-processing
step has a total weight of 1, since each such bin contains a large
item (that has a weight of 1) and a small item of weight 0. Thus
$\sum\limits_{i\notin I'}w_2(i) =X$, and in total
$\sum\limits_{i\in I}w_2(i) \geq X+Y-3=f_1(\mh)-3$.

Next, we define a weight measure $w_1$. Consider the $t$ large
items, and their packing in an optimal solution \opt. For any
large item $a$, which is packed in a bin with at least one other
(small) item, consider the largest small item which is packed with
$a$ and denote it by $z_a$. If $z_a$ is not well-defined, one of
the possible items is chosen arbitrarily to be defined as $z_a$.
If no such item exists, i.e., $a$ is packed as a single item in a
bin of \opt, we add an item of size zero to this bin of \opt\ and
define it to be $z_a$. Therefore $z_a$ exists and is defined
uniquely for every large item $a$. We define the weight of every
item $i$ as $w_1(i)=w(s_i)$, except for the items $z_a$ for
$a=1,\ldots,t$, for which we let $w_1(z_a)=\frac{w(s_{z_a})}{2}$.

In order to show $W_2(I) \leq W_1(I)$, we define a valid matching
in the auxiliary graph. This matching is based on the packing of
\opt. Let $Z=\{z_a| 1 \leq a \leq t \}$ and denote a set of the
largest $\lceil \frac t2 \rceil$ items in $Z=\{z_a| 1 \leq a \leq
t \}$ by $Z'$. We initialize the matching with the items of $Z'$
being matched to the large items from their bins in \opt. This
matching is valid since by definition of $Z$, each item in this
set is packed in \opt\ in a different bin, with a different large
item. If the $\lceil \frac t2 \rceil$ items matched to them are
not exactly items $\lceil \frac {t+1}2 \rceil,\ldots,t$, it is
possible to replace some large items in the matching by smaller
large items, until this situation is reached. We have $s_{i_1}
\leq s_{i_2}$ for $i_1 \in Z\setminus Z'$ and $i_2 \in Z'$. Since
the function $w$ is monotonically non-decreasing, we get
$\sum\limits_{z_a \in Z} w(s_{z_a}) \leq 2 \sum\limits_{z_a \in
Z'} w(s_{z_a})$. Let $W(I)=\sum\limits_{i=1}^n w(s_i)$. We have
$W_2(I)=W(I)-c(M)$, where $c(M)$ is the cost of a matching in the
auxiliary graph, with a maximum cost, and
$W_1(I)=W(I)-\sum\limits_{1 \leq a \leq t} \frac {w(s_{z_a})}2
\geq W(I)-\sum\limits_{z_a \in Z'} w(s_{z_a})\geq
W(I)-c(M)=W_2(I)$, since $c(M)$ is a maximum cost matching on the
smallest $\lceil \frac t2 \rceil$ large items, and
$\sum\limits_{\lceil \frac {t+1}2 \rceil \leq a \leq t}
w(s_{z_a})$ is the cost of one such matching, which we defined
above.

Finally, we need to find an upper bound on the total weight in a
bin of \opt, according to $w_1$. We first consider bins that do
not contain a large item. For any item $i$ of size $s_i=\beta \in
(0,\frac 12]$, we have $w_1(i)\leq \frac 32 \beta$. For items of
size 0 the weight is 0. Therefore, the total weight of items in
such a bin is no larger than 1.5 (a tighter upper bound of 1.423
is proved in \cite{BC81}).

Consider next a bin which contains a large item. Let $a$ be the
large item of this bin, and $z_a$ is chosen as above. If
$s_{z_a}=0$, then the only item in the bin that has a non-zero
weight according to $w_1$ is $a$, and thus the total weight is 1.
Otherwise, let $j$ be such that $s_{z_a}\in (\frac 1{j+1},\frac
1j]$. Any other item $i$ in the bin (except for $a$ and $z_a$)
satisfies $w_2(i)\leq \frac{j+1}j s_i$ (since $s_i \leq s_{z_a}
\leq \frac 1j$). If $j = \pi_i-1$ for some $i\geq 1$, we have
$w_2(z_a)=\frac 1{2j}$. Otherwise,
$w_2(z_a)=\frac{j+1}{2j}s_{z_a}$.

We have a total weight of at most
$1+w_2(z_a)+\frac{j+1}{j}(1-s_a-s_{z_a})\leq
1+w_2(z_a)+\frac{j+1}{j}(\frac 12-s_{z_a})$, since $s_a>\frac 12$.
In the first case we use $s_{z_a} > \frac{1}{j+1}$, and get at
most $1+\frac 1{2j}+\frac{j+1}{2j}-\frac{1}{j}=\frac 32$. In the
second case we get at most
$1+\frac{j+1}{2j}s_{z_a}+\frac{j+1}{2j}-\frac{j+1}{j}s_{z_a}
=\frac{3j+1}{2j}-\frac{j+1}{2j}s_{z_a}$. Using the same property
we get at most $\frac 32$ again.
\end{proof}

Next, we perform an analysis for functions $f_k$ with $k\geq 2$.
Let $I$ be the original input on which \mh\ is executed. Let
$\hat{I}$ denote an input in which every small item, which  is
matched with a large item in the pre-processing step of \mh, is
replaced with an item of size $s_1$. Thus, at most $\lceil \frac
t2 \rceil$ items are increased to the size $s_1$. We consider the
following solutions and compare their costs. The cost of the
solution of \mh\ on $I$, with respect to $f_k$, is denoted by
$A_k(I)$. The cost of the solution of \nfi\ on $\hat{I}$,  with
respect to $f_k$, is denoted by $\nfi_k(\hat{I})$. The next
solution that we consider is an overflowed solution that is
created for $I$ as follows. The items are sorted by size in a
non-decreasing order (that is, order by indices in a decreasing
order). At each time, a minimum prefix of the items of total size
larger than 1 is assigned to the next bin. The cost of this
solution with respect to $f_k$ is denoted by $O_k(I)$. The cost
of an optimal solution for $I$, with respect to $f_k$, is denoted
by $\opt_k(I)$. Finally, we consider a solution for $\hat{I}$
with consecutive bins, which is constructed from the overflowed
solution for $I$ as follows (the construction is similar to the
one in \cite{ABSL}, except for the treatment of items in
$\hat{I}$, and the fact that the corresponding items in $I$ are
simply removed). For every bin of the overflowed solution, if the
total size of items exceeds 1 (this is the case with all bins
except for possibly the last bin, or bins with removed items),
remove the last item and open a new bin for it. The additional
large items of $\hat{I}$, which existed as smaller items in $I$
and were removed from $I$, are assigned to dedicated bins.  The
cost of this solution, with respect to $f_k$ is denoted by
$C_k(\hat{I})$.

By Lemma \ref{consec}, we have $\nfi_k(\hat{I}) \leq
C_k(\hat{I})$. By Lemma \ref{overflow}, we have $O_k(I)\leq
\opt_k(I)$. We next prove two lemmas after which we will be able
to conclude $A_k(I) \leq \frac 32 \opt_k(I)+3.5$.

\begin{lemma}
$A_k(I) \leq \nfi_k(\hat{I})+1$. \end{lemma}
\begin{proof}
Since all small items of $I$ that are packed in the
pre-processing step of \mh\ are large in $\hat{I}$, the small
items packed by \nfi\ in the two algorithms are the same ones, and
bins created by \nfi\ in the two algorithms are identical, except
for bins that contain a large item. If any of the two
applications of \nfi\ outputs a bin that contains a large item
together with other items, we adapt the solution by moving this
item into a separate bin, this modification cannot decrease the
cost of a solution, but it may increase the cost by at most 1.
The small items bins, resulting from running \nfi\ in both
solutions (the solution of \mh\ and the solution of \nfi,
possibly with the modification) are now identical. The remaining
items are packed in both solutions either in singles or in pairs.
Thus the costs of such bins are equal in both solutions (since $k
\geq 2$). Therefore, the claim follows.
\end{proof}

\begin{lemma}
$C_k(\hat{I}) \leq \frac 32 O_k(I)+2.5$.
\end{lemma}
\begin{proof}
We first modify both solutions so that none of them combines
large items with some small item in one bin (but the overflowed
solution may still have bins with two large items, which are not
modified here). For the overflowed solution, this may require
moving one or two large items from a shared bin to a dedicated
bin, so it may increase the cost by at most 2. For the other
solution, this may involve moving one large item to a dedicated
bin, and cannot decrease the cost of the solution. We consider
first the bins with small items, that contain at least $k+1$
items in the overflowed solution. For every such bin, its cost is
at least $k$. As a result of moving the last item to a dedicated
bin (in the process of creation of the feasible solution), an
additional cost of at most 1 is incurred. Thus the cost increases
by at most a factor of $\frac 32$. For any bin containing at most
$k$ items, there is no additional cost from this step. Note that
all bins with large items are in this situation. The cost of bins
with large items in the overflowed solution with the modification
is simply $t$, no matter how they are exactly packed, so packing
each one in a dedicated bin does not change the cost. Together
with the additional $\lceil \frac t2 \rceil$ large items, the
cost of large items becomes $\lceil \frac {3t}2 \rceil \leq
\frac{3t}{2}+\frac 12$. Removing small items that do not exist in
$\hat{I}$ may only decrease the cost. This proves the claim.
\end{proof}

 Using Lemma \ref{reduce}, we have proved the following.
\begin{theorem}
The asymptotic approximation ratio of \mh\ is at most 1.5. for any
non-decreasing concave function $f$ with $f(0)=0$.
\end{theorem}

We have shown above that for $k=2$ (and similarly, for any
constant $k$), the bound 1.5 is tight. Note that the bound 1.5 is
tight for $k=1$ as well. Consider an input with $N$ large items of
size $\frac 12 +\frac 1{2K}$, and $N(K-1)$ small items of size
$\frac{1}{2K}$ (for large enough $N,K$, such that $N$ is divisible
by $4K$). \mh\ creates $\frac N2$ bins with one large and one
small item, $\frac{N(K-1)-\frac N2}{2K}=\frac{N(K-\frac 32)}{2K}$
bins with $2K$ small items each, and $\frac{N}{2}$ bins with one
large item. This gives a total cost of $
N+\frac{N(K-\frac32)}{2K}$. An optimal solution combines $K-1$
small items with every large item, for a cost of $N$. For large
enough $K$, the ratio is arbitrarily close to $1.5$. It can be
seen that this ratio is achieved for any fraction $0 \leq \alpha
\leq 1$ of large items that participate in the pre-processing
step.

\section{An AFPTAS for \gc}
\label{afptas} In this section we present our main result, that
is, an AFPTAS for \gc. We give a sketch which presents the main
ideas and technical difficulties, and give the full description
of the AFPTAS and its analysis later. We first present an
auxiliary algorithm called {\it Fractional Next-Fit Increasing}.

\subsection{The analysis of \fnfi}
We prove a property which is helpful in the design of our AFPTAS.
It is related to the property on \nfi\ in Lemma \ref{consec}, but
it is stronger since it is proved for any non-decreasing concave
function $f$ with $f(0)=0$, for fractional packing of items. A
packing is fractional if items can be cut into pieces, where
pieces of one item can possibly be packed in different bins. We
assume without loss of generality that in every fractional
packing, every bin contains at most one part of each item. If this
property does not hold, it is possible to unite parts of items
within a bin without changing the cost.

We consider an algorithm which creates a fractional packing  of
the items according to the variant of the \nfi\ heuristic, called
{\sc Fractional} \nfi\ (\fnfi). This algorithm sorts items by size
in non-decreasing order. At each time, a bin is filled completely,
before moving on to the next bin. For this, we allow the splitting
of items into several parts, that is, the last item that is packed
in a bin is possibly just a part of an item. Consequently, the
first item packed in the next bin may be the remaining part of the
same item.  Note that each bin in the output of \fnfi\ contains at
most two split items and that in total only at most $m-1$ items
are split (where $m$ is the number of bins used by \fnfi).

Note that there is no advantage in packing fractions of size zero
of items, except for zero sized items, which we assume that are
split between bins. If a part of size $\alpha$ of an item of size
$\beta>0$ is packed in a given bin, we say that the fraction of
this item that is packed in this bin is $\frac{\alpha}{\beta}$. If
an item is packed in a bin completely, we say that its fraction
packed in the bin is 1. The number of items in a bin which is
packed fractionally is the sum of fractions in it. This number is
not necessarily an integer and it is unrelated to sizes of these
fractional items, but only to their fractions.

To be able to analyze fractional packings, we next define $f$ for
any (real and not necessarily integral) value $q \in [0,n]$ as
follows. We define $f(q)$, for $i < q < i+1$, to be $(i+1-q)\cdot
f(i)+(q-i)\cdot f(i+1)$. The values of $f$ for integer values of
$q$ are unchanged. We let $f(x)=f(n)$ for any $x \geq n$. This
function is piecewise linear and continuous, and since it is an
extension of a non-decreasing concave function on integers, it is
monotonically non-decreasing and concave in $[0,n]$. The cost of a
fractional packing  is calculated according to the {\it
generalized} function $f$, using the numbers of items packed into
the bins as defined above.

A simple property of \fnfi\ is that it creates bins that are
sorted in a non-increasing order of the number of items in them.
This holds since given two bins $i_1<i_2$, bin $i_1$ is completely
occupied, and every item that has a part packed in bin $i_1$ has a
size no larger than any item that has a part packed in bin $i_2$.

For any non-decreasing concave function $f$ with $f(0)=0$, the
following lemma states that \fnfi\ is the ``best'' heuristic among
packings with fractionally packed bins. Consider a given input
$I$, a cost function $f$ and a fractional packing, $\B$.


\begin{lemma} \label{consecf}
$f(\fnfi(I))\leq f(\B(I))$.
\end{lemma}
\begin{proof}
Assume by contradiction that for an input $I$, a fractional
packing $\B$ and a function $f$, we have $f(\fnfi(I)) > f(\B(I))$.
Assume that the bins of $\B$ are sorted according to a
non-increasing numbers of items. If the packing $\B$ that
satisfies the condition is not unique, consider such a packing
$\B$ which maximizes the suffix of bins that are packed
identically to the packing of \fnfi. Consider the first bin $i$ of
$\B$ that is packed differently from the packing of \fnfi. If bin
$i$ is the very last bin of the packing $\B$, then the bins
$1,\ldots,i-1$ are packed as in the packing of \fnfi, and
therefore, bin $i$ also has the same contents for $\B$ as it has
for \fnfi. Therefore we assume that $i$ is not the last bin of
$\B$.

Let $j,j+1,\ldots,j'$ be the indices of items that \fnfi\ packs in
bin $i$ (the first and last items, which have the indices $j'$ and
$j$ respectively, may be packed fractionally in this bin). Let $j
\leq j_1 \leq j'$ be an index of an item such that $\B$ packs a
smaller part of $j_1$ (possibly of size zero) in bin $i$ than
\fnfi\ does. Such an item must exist by the following argument. If
\fnfi\ fills bin $i$ completely, then since bin $i$ of $\B$ is
packed differently, it cannot have at least the same fraction of
every item. Otherwise, \fnfi\ packs all the remaining items in bin
$i$, so a different packing of bin $i$ means that some item has a
smaller fraction in $\B$.

We next consider the case that there exists an item $j_2$ for
which $\B$ packs a larger part in bin $i$ than the packing of
\fnfi. Since the two algorithms pack bins $1,\ldots,i-1$
identically, only the items of index up to $j'$  are available for
packing in bins $i,i+1,\ldots$, where the item of index $j'$ may
already be fractional. Out of these items, \fnfi\ packs a maximum
prefix into bin $i$, so this item must satisfy $j_2 \leq j$. We
get that $j_2 \leq j \leq j_1$. Since $j_1 \neq j_2$ by their
definitions, we get $j_2 < j_1$.

Denote the fractions of $j_1$ and $j_2$ in bin $i$ of $\B$ by
$\gamma_1$ and $\gamma_2$, and the fractions of $j_1$ and $j_2$ in
bin $i$ of \fnfi\ by $\delta_1$ and $\delta_2$. We have $\delta_1
> \gamma_1  \geq 0$ and $\gamma_2
> \delta_2 \geq 0$. Since $\gamma_1<\delta_1$, and bins $1,\ldots,i-1$ are packed identically
in both algorithms, there exists a further bin $i'$ that contains
a part of item $j_1$ in the packing of $\B$. Let $\eps_1>0$ be the
 fraction of $j_1$ in bin $i'$ of $\B$.

We would like to swap parts of items in the packing of $\B$,
specifically, a part of item $j_1$ from bin $i'$ with a part of
item $j_2$ in bin $i$. We use $\mu$ to denote the size of the
swapped part. There are three restrictions on $\mu$. The resulting
fraction of $j_1$ in bin $i$ of $\B$ cannot exceed the fraction of
this item in bin $i$ of \fnfi, thus $\mu \leq
(\delta_1-\gamma_1)s_{j_1}$. We can swap at most a fraction
$\eps_1$ of $j_1$. Moreover, we can swap at most a fraction of
$\gamma_2-\delta_2$ of $j_2$, in order to keep a fraction of $j_2$
in bin $i$ that is at least as large as the one in bin $i$ of
\fnfi. Therefore, we let
$\mu=\min\{(\gamma_1-\delta_1)s_{j_1},(\gamma_2-\delta_2)
s_{j_2},\eps_1 s_{j_1} \}$. We adapt $\B$ by swapping a part of
size $\mu$ of item $j_1$ from bin $i'$ with a part of size $\mu$
from $j_2$ in bin $i$. By definition of all variables, $\mu>0$,
and thus some change occurred.

Let $n_i$ and $n_{i'}$ be the original numbers of items in bins
$i$ and $i'$ of $\B$. By our assumption $n_i \geq n_{i'}$. Let
$\alpha_1$ and $\alpha_2$ be the fractions of items $j_1$ and
$j_2$ that are swapped. Since $\mu=\alpha_1 \cdot s_{j_1}=\alpha_2
\cdot s_{j_2}$, and $s_{j_1} \leq s_{j_2}$, we have $\alpha_1 \geq
\alpha_2$. Thus, the change in the cost is
$f(n_i-\alpha_2+\alpha_1)+f(n_{i'}-\alpha_1+\alpha_2)-f(n_i)+f(n_{i'})\leq
0 $, by concavity. As a result of this process, the total number
of items in bin $i$ remains no smaller than the numbers of items
in each of the bins $i+1,i+2,\ldots$.

If an item $j_2$ does not exist, it means that bin $i$ has a total
size of items that is smaller than the total size of items in bin
$i$ of \fnfi. In particular, it means that bin $i$ is not fully
packed. We define $\gamma_1$, $\delta_1$, $i'$ and $\eps_1$
as before. 
In this case we can define
$\mu=\min\{(\gamma_1-\delta_1)s_{j_1},\eps_1 s_{j_1}\}$. We define
$\alpha_1$, $n_i$ and $n_{i'}$ as before. Thus, the change in the
cost is $f(n_i+\alpha_1)+f(n_{i'}-\alpha_1)-f(n_i)+f(n_{i'})\leq 0
$, by concavity.

It is possible to perform this process on bin $i$ multiple times,
until there is no item that has an item for which a smaller
fraction of it is packed in bin $i$ of $\B$ than it is packed in
the same bin for \fnfi. At this time these bins become identically
packed.


We next show that this situation, where no item $j_1$ exists, is
reached after a finite number of swaps. For every item $j_1$, it
can be performed for every item $j_2$ and for every successive
bin. This gives a total of at most $n^3$ swaps, and possibly $n^2$
movements of items to bin $i$ without swaps.

After we reach the situation where bin $i$ is identical for $\B$
and \fnfi, the bins $1,\ldots,i$ of $\B$  are sorted by a
non-increasing number of items. Each remaining bin of $\B$ has a
number of items that is no larger than bin $i$. Moreover, bins
$i+1,i+2,\ldots$ can be sorted so that the list of bins becomes
sorted as required. The changes above can only decrease the cost
of the solution, and therefore we get a contradiction to our
assumption.
\end{proof}

\subsection{The sketch of the scheme}
We define an item to be a small item if its size is  smaller than
$\eps$ and otherwise it is a large item.  Denote by $S$ the set of
small items and by $L$ the set of large items.  Our first step is
to apply linear grouping \cite{FerLue81} of the large items, that
is we sort them by size and we partition them into $1\over \eps^3$
(almost) equal-sized sets of consecutive items (in the sorted
list).  We pack each item of the set of the largest items in its
own bin, and we round-up the size of the items in each other set
to the largest size of an item in its set.

We next partition the items in $S$ into $S'\cup S''$  where $S''$
contains the smallest items such that the total size of the items
in $S''$ is close to a constant which we define depending on
$\eps$. The items of $S''$ are packed nearly optimally using the
\fnfi\ heuristic and packing any split item using a dedicated bin.
These bins will enable us to use a constant number of bins with an
arbitrary content (of items in $L\cup S'$) while paying at most
$\eps$ times the cost of the bins which are used to pack the items
in $S''$. We note that packing $S''$ using the \nfi\ heuristic is
also possible and leads to a similar performance guarantee.
However, the analysis of using \fnfi\ is simpler.

Our next step is to approximate the cost function $f$ using a
staircase (step) function with $O(\log f(n))$ steps.  We use
concavity of $f$ to show that this number of steps in the function
is sufficient to get a $1+\eps$ approximation of $f$.

We next move on to finding a packing of the  items in $L\cup S'$
(neglecting the largest items which are packed in dedicated bins).
In such an instance, the linear program, which we construct,
allows the small items of $S'$ to be packed fractionally.  To
construct this linear program we define a set of configurations of
large items (this is the standard definition), and a set of
extended configurations which also define the space and
cardinality of small items in a configuration (this is a
non-standard idea).  The linear program will decide how many bins
with a given extended configuration to open and what type of bins
each small item need to be packed in.  These types are called
windows, and we define them as the pair consisting of the total
space for the small items and the total cardinality of small items
in a bin with this window.  Hence in this linear program we have a
constraint for each size of large items (a constant number of
constraints) a constraint of each small item (a linear number of
such constraints), and two constraints for each type of windows.
We apply the column generation technique of Karmarkar and Karp
\cite{KK82} to solve approximately the resulting linear program
(we use a separation oracle which applies an FPTAS for the
Knapsack problem with cardinality constraint given by
\cite{CKPP}).

Unfortunately the number of fractional  entries in a basic
solution for this linear program (as we can assume our solution is
indeed a basic solution), is linear in the number of windows types
(plus a constant).  The number of windows is indeed polynomial in
the input size allowing us to solve the linear program, but it is
not a constant, and we will incur a too large error if we would
like to round up the fractional solution.

Hence, we define a restricted set of windows types with a much smaller set of windows, and we show how to project cleverly our solution to a new solution which is not worse than the original solution, whose support uses only windows from this restricted set of windows.  Therefore, when we count the number of constraints, we can eliminate the constraints corresponding to windows which do not belong to the  restricted set of windows.  Thus the new bound on the number of fractional components in the projected solution is now much smaller.  That is, our projected solution which is an approximated solution to the original linear program is also an approximated solution to the linear program with additional constraints setting the variables to zero if the corresponding window does not belong to the restricted set of windows.

The next step is  to round up the resulting projected solution. If
a small item is packed fractionally, then we pack it in its own
dedicated bins. If the fractional solution needs to pack
fractional copies of bins with a given extended configuration,
then we round up the number of such bins.  The large items clearly
can be packed in these bins according to the configurations of the
large items.  The small items are now assigned to windows (by an
integral assignment), and not to specific bins.  Therefore, our
last steps are devoted to packing the small items.

We first place the small items which are packed in a common window
type into the bins with this window as part of their extended
configuration in a round-robin fashion where the small items are
sorted according to their size  (this ensures us that the number
of items in each such bin will be approximately the same, and the
total size of these items in such bins will be approximately the
same).  Hence, the excess of volume of small items in a bin is
relatively small (with respect to the total size of small items in
this bin).  In fact it is at most one excess item per bin plus a
small volume of additional small items (this small volume is due
to a rounding we have done when we define the set of windows). The
excess items are packed in dedicated bins such that $1\over \eps$
excess items are packed in each dedicated bin.  The small volume
items are packed again in dedicated bins such that these items
from $1\over \eps$ bins are packed into one common dedicated bin.
The items which are removed from a bin after the process of the
round-robin allocation are the largest small items of this given
excess volume.    The resulting scheme is an AFPTAS for \gc, as
claimed by the following theorem.

\begin{theorem}\label{afptasthm}
The above scheme is an AFPTAS for \gc.
\end{theorem}

\subsection{A detailed description and analysis of the AFPTAS for
\gc} \label{theafptas} Let $0<\eps \leq \frac 13$ be such that
$1\over \eps$ is an integer. Recall that $f(0)=0$ and $f(1)=1$.

The input for this problem includes in addition to the list of
items, also the function $f$. Therefore, the running time needs to
be polynomial in the following four parameters: $n$,
$\frac{1}{\eps}$, and the binary representations of the numbers in
the input, including the item sizes, and the values of $f$ on the
integers $1,\ldots,n$. The length of the representation of $f$ is
at least $\log f(n)$.

If $n \leq \frac{1}{\eps}$, we pack each item into a separate bin.
In this case, the cost of the solution is at most
$\frac{f(1)}{\eps} \leq  (1+\eps)\opt+\frac{1}{\eps}$. We
therefore assume that $n > \frac{1}{\eps}$.

{\bf Linear grouping.} An item $j$ is {\it large} if $s_j\geq
\eps$. All other items are {\it small}. We denote by $L$ the set
of large items, and by $S$ the set of small items. We perform
linear grouping of the large items. That is,  if $|L| \geq {1\over
\eps^3}$, then for $m={1\over \eps^3}$ we partition $L$ into $m$
classes $L_1,\ldots ,L_m$ such that $\lceil |L|\eps^3\rceil
=|L_1|\geq |L_2|\geq \cdots \geq |L_m|=\lfloor |L| \eps^3
\rfloor$, and $L_p$ receives the largest items from $L \setminus
\left[ L_1\cup \cdots \cup L_{p-1}\right]$).  The two conditions
uniquely define the allocation of items into classes up to the
allocation of equal size items.  For every $j=2,3,\ldots ,m$ we
round up the size of the elements of $L_j$ to the largest size of
an element of $L_j$. For an item $i$, we denote by $s'_{i}$ the
rounded-up size of the item.  If $|L| < {1\over \eps^3}$, then
each large item has its own set $L_i$ such that $L_1$ is an empty
set, and for a large item $j$ we let $s'_j=s_j$ (i.e., we do not
apply rounding in this case). In both cases we have $|L_1| \leq
2\eps^3 |L|$.

For items in $L_1$, we do not round the sizes, and we denote
$s'_j=s_j$ for all $j\in L_1$.  For $j\in S$ we also let
$s'_j=s_j$. We denote by $L'=L\setminus L_1$. We consider the
instance $I'$ consisting of the items in $L'\cup S$ with the
(rounded-up) sizes $s'$. Then, using the standard arguments of
linear grouping we conclude $\opt(I') \leq \opt(I)$.
 The items in
$L_1$ are packed each in a separate bin. We next describe the
packing of the items in $I'$.

{\bf Dealing with the set of the smallest items.}  We define a
partition of the set $S$ into two parts $S'$ and $S''$, such that
$S''$ is a suffix of the list of input items (i.e., a set of
smallest items). Specifically, if $i\in S'$ and $j\in S''$, then
$s'_i \geq s'_j$. Let $S''$ be a maximum suffix $\{p,\ldots,n\}$,
such that $S'' \subseteq S$, for which the total size is at most
$1+h(\eps)$, where $h(\eps)$ is a function of $\eps$ that we will
define later. This function is defined such that $h(\eps)\geq
{1\over \eps}$ is an integer for any valid choice of $\eps$. Note
that if the total size of the small items is smaller than
$1+h(\eps)$ then we let $S''=S$ and $S'=\emptyset$. We will pack
the items from $S''$ independently from other items.  That is,
there are no mixed bins containing as items from $S''$ as items
not from $S''$.

The first packing step of the algorithm is to pack the items of
$S''$ using the following heuristic. We apply \fnfi\ (processing
the items in an order which is reverse to their order in the
input). This results in $1+h(\eps)$ bins, unless $S''=S$.
Afterwards, a new dedicated bin is used for every item that was
split between two bins by \fnfi. There are at most $h(\eps)$ such
items.

In order to focus on solutions that pack the items of $S''$ as we
do, we next bound the cost of a solution that packs the items in
$S''$ in this exact way (packed by \fnfi\ in separate bins, where
split items are moved to an additional bin). On the other hand, we
relax our requirements of a solution and allow fractional packing
of the items in $S'$. The solution clearly needs to pack the items
in $L'$ as well (no fractional packing can be allowed for large
items). We denote the optimal cost of such a solution by
$\opt'(I')$. The motivation for allowing fractional packing of the
items of $S'$ is that our goal is to bound the cost of solutions
to a linear program that we introduce later, and this linear
program allows fractional packing of small items that are
considered by it, which are exactly the items of $S'$ (while the
items of $S''$ remain packed as defined above).

\begin{lemma}
$\opt'(I') \leq (1+\eps)\opt(I') + (3h(\eps)+3) \cdot f({1\over
\eps})\leq (1+\eps)\opt + (3h(\eps)+3) \cdot f({1\over \eps})$.
\end{lemma}
\begin{proof}
Consider an optimal solution $\overline{\opt}$ to the following
relaxation \Prob\ of our packing problem. We need to pack the
items of $I'$ (with rounded-up sizes) but {\it all} the items of
$S$ can be packed fractionally.  The difference with the packing
$\opt'(I')$ is that items of $S''$ can be packed in an arbitrary
way, and not necessarily into dedicated bins, as is described
above. In particular, they can be packed fractionally. The
difference with the packing $\opt(I')$ is the possibility to pack
the small items fractionally. The cost of $\overline{\opt}$ is
clearly at most $\opt(I') \leq \opt$.

We sort the bins of $\overline{\opt}$ in a non-increasing order,
according to the number of items (i.e., the sum of fractions of
items) packed in the bin (including large items). Let $\sigma_i$
be the total free space in bin $i$ that is left after packing its
large items in it. This is the space which is used by small items,
together with all the free space, if exists. Let
$\Sigma_i=\sum\limits_{j=1}^i \sigma_i$. Let $p=\min\{i| \Sigma_i
\geq \sum\limits_{j \in S''} s'_j\}$. The integer $p$ must exist
since all items of $S''$ must be packed.

We show that without loss of generality, we can assume that all
items of $S''$ are packed in bins $1,2,\ldots,p$ in
$\overline{\opt}$. To show this, consider an optimal solution to
\Prob\ that minimizes the following function (among all optimal
solutions): the number of existing quadruples $(a_1,i_1,a_2,i_2)$,
where $a_1 \leq p < a_2$, $i_1\in S'$, $i_2 \in S''$, and there is
a non-zero fraction of item $i_j$ packed in bin $a_j$, for
$j=1,2$. Assume by contradiction that such a quadruple
$(a_1,i_1,a_2,i_2)$ exists. Let $\gamma$ be the fraction of $i_1$
in bin $a_1$ and $\delta$ the fraction of $i_2$ in bin $a_2$.

Let $\mu = \min\{\gamma\cdot s_{i_1},\delta\cdot s_{i_2}\}$.
Denote the fractions of $i_1$ and $i_2$ of size $\mu$ by
$\gamma'=\frac{\mu}{s_{i_1}}$ and $\delta'=\frac{\mu}{s_{i_2}}$.
We swap a part of size $\mu$ of item $i_2$ in bin $a_2$ with a
part of size $\mu$ of item $i_1$ in bin $a_1$. Since $s_{i_1} \geq
s_{i_2}$ (recall that $S''$ contains the smallest items), we get
that the fractions satisfy $\gamma' \leq \delta'$. The number of
items in bin $a_1$ was changed by $\delta' - \gamma'$, and in bin
$a_2$ it was changed by $\gamma' - \delta'$. The sorted order of
bins may have changed as a result, but bin $a_1$ can be moved to
an earlier spot while $a_2$ may be moved to a later spot, so the
set of the first $p$ bins does not change. Moreover, we destroyed
at least one quadruple, and did not create new ones, since no
parts of items of $S'$ were moved to bins $1,\ldots,p$ and no
items of $S''$ were moved to bins $p+1,p+2,\ldots$. Let $n_1$ and
$n_2$ be the numbers of items in bins $a_1$ and $a_2$ before the
change. The change in the cost function is $f(n_1+\delta' -
\gamma')+f(n_2-(\delta' - \gamma'))-f(n_1)-f(n_2)\geq 0$, since
$n_1\geq n_2$, $\delta' - \gamma' \geq 0$, and by concavity.
Therefore, the resulting solution has a cost of at most
$\overline{\opt}$, and the minimality is contradicted.

If no such quadruple exists then there are two cases. If all bins
$p+1,p+2,\ldots$ contain only fractions of items of $S'$ (possibly
in addition to large items), then all items of $S''$ are in bins
$1,\ldots,p$ and our assumption holds. Otherwise, we have that all
bins $1,\ldots,p$ contain no fractions of items in $S'$. In this
case, if there are items of $S''$ in any of the bins
$p+1,p+2,\ldots$, then there must be empty space in bins
$1,\ldots,p$. Parts of items of $S''$ can be repeatedly moved to
these bins, until no parts of items of $S''$ exist in bins
$p+1,p+2,\ldots$. In each such step, the number of items in some
bin in $1,\ldots,p$ increases, and the number of items in some bin
in $p+1,p+2,\ldots$ decreases. Sorting the bins again after every
such step (according to a non-increasing numbers of items) will
contain the same set of bins in the prefix of $p$ bins, and our
assumption holds as well. Due to concavity, and since the target
bin cannot contain less items than the source bin, every such step
cannot increase the cost.

We next adapt $\overline{\opt}$ by creating at most $h(\eps)+2$
additional bins, and move the small items of the first $p$ bins
into these bins using \fnfi\ (that is, the list of items is
processed in a reverse order from their order in the input and
packed fractionally into bins). Note that this set of small items
may contain items of $S'$ of total size at most 1 (out of these
items of $S'$, at most one is split between two bins), and the
total size of items of $S''$ is at most $h(\eps)+1$. We denote
this set of items that is moved by $\hat{S}$. We compute the
change in the cost and afterwards adapt the solution further so
that it complies with the requirement that the items of $S''$ are
packed integrally in separate bins, as is done above.

We define an auxiliary monotonically non-decreasing concave
function $\tilde{f}$ as follows. $\tilde{f}(x)=f(x+\frac
1{\eps})-f(\frac{1}{\eps})$. Note that $\tilde{f}(0)=0$. Consider
the $p$ bins of $\overline{\opt}$ from which the small items are
removed. Let $r_i$ and $a_i$ denote the numbers of large and small
items in these original bins. Clearly, $r_i \leq \frac{1}{\eps}$.
By removing the small items, the cost of such a bin decreases by
$f(a_i+r_i)-f(r_i)\geq
f(a_i+\frac{1}{\eps})-f(\frac{1}{\eps})=\tilde{f}(a_i)$, where the
inequality is due to concavity. For every bin which is created for
small items, if it contains $b_i$ small items, its cost is
$f(b_i)\leq f(b_i+\frac
1{\eps})=\tilde{f}(b_i)+f(\frac{1}{\eps})$, where the inequality
is due to monotonicity.

Consider now the packing of the items $\hat{S}$ that is implied by
the solution $\overline{\opt}$, with respect to the function
$\tilde{f}$, and neglecting the large items. The cost of this
packing for bin $i$ is $\tilde{f}(a_i)$. Let $\tilde{A}$ denote
the total cost of all the bins that contain items of $\hat{S}$,
that is, of the first $p$ bins. Let $\tilde{B}$ denote the total
cost with respect to $\tilde{f}$ of all the bins that are created
by \fnfi\ for $\hat{S}$. In this case the cost of a bin $i$ is
$\tilde{f}(b_i)$. That is, $\tilde{A}=\sum\limits_{i=1}^p
\tilde{f}(a_i)$ and $\tilde{B}=\sum\limits_{i=1}^{h(\eps)+2}
\tilde{f}(b_i)$. By Lemma \ref{consecf} (that holds even though
the value $\tilde{f}(1)$ can be arbitrary), we have $\tilde{A}
\geq \tilde{B}$.

Let $\Delta$ denote the difference in the cost for the items of
$\hat{S}$. We have $\Delta=\sum\limits_{i=1}^{h(\eps)+2} f(b_i)
-\sum\limits_{i=1}^p(f(r_i+a_i)-f(r_i))\leq
\sum\limits_{i=1}^{h(\eps)+2} (\tilde{f}(b_i)+f(\frac{1}{\eps}))
-\sum\limits_{i=1}^p\tilde{f}(a_i)\leq (h(\eps)+2)f(\frac
1{\eps})$ (by the previous claims and $\tilde{A} \geq \tilde{B}$).

We next convert the packing of small items as follows. If there
exists a mixed bin, that is, a bin containing items from both
$S''$ and $S'$, we split it into two bins, so that the two subsets
of $S'$ and of $S''$ are separated. If a mixed bin indeed exists,
$S' \neq \emptyset$, and the total size of the $S''$ items is more
than $h(\eps)$, but not more than $h(\eps)+1$. Therefore, the
split bin appears as the $h(\eps)+1$-th bin created by \fnfi.
Moreover, the number of items in the $h(\eps)+1$-th bin is no
larger than the number of items in every earlier bin. Therefore,
if the number of items in the $h(\eps)+1$-th bin is $N$, then the
current cost is at least $f(N)(h(\eps)+1)$ and as a result of the
split, the cost increases by an additive factor of at most $f(N)$.
So the multiplicative factor of the increase in the cost is at
most $1+{1\over h(\eps)+1} \leq 1+\eps$ where the inequality holds
by $h(\eps) \geq {1\over \eps}$.

For a pair of consecutive bins created by \fnfi\ (excluding the
bins with items of $\hat{S}\cap S'$), if an item was split between
the two bins, it is removed from these bins and packed completely
in a new bin dedicated to it. There are at most $h(\eps)$ such
items so this increases the cost by at most $h(\eps)\cdot f(1)$.
At this time, the items of $S''$ are packed exactly as in
$\opt'(I')$.

The total cost is at most
$(1+\eps)(\overline{\opt}+(h(\eps)+2)f(\frac 1{\eps}))+h(\eps)\leq
(1+\eps)(\overline{\opt})+((2+\eps)h(\eps)+2+2\eps)
f(\frac{1}{\eps})\leq (1+\eps)(\overline{\opt})+(3h(\eps)+3)
f(\frac{1}{\eps})$ (using $\eps \leq \frac 13$).
\end{proof}

We next need to pack the items in $I''=I'\setminus S''$. Let
${\updelta} = \min\limits_{i \in S'} s'_i$. Clearly, for any $i
\in S''$ we have $s'_i\leq \updelta$. Let
$\bf{\Delta}=\frac{1}{\updelta}$.

We next consider the instance $I''$. In the temporary solutions,
we allow fractional packing of the items of $S'$ and we use
$\opt(I'')$ to denote an optimal packing of $I''$ where small
items may be packed fractionally. This does not change the fact
that any bin, packed with items of a total size of at most 1, can
contain a total number of items of at most $\bf \Delta$ even if it
contains fractions of items.

We denote the cost of the bins packed with the items of $S''$ by
$F(S'')$. By definition we have $\opt'(I')=\opt(I'')+F(S'')$. The
items of $S''$, if packed by \fnfi\ (which by Lemma \ref{consecf}
is a minimum cost packing for them) require at least $h(\eps)$
full bins, with at least $\bf \Delta$ items in each. Therefore, we
have $F(S'') \geq h(\eps)\cdot f(\bf \Delta)$. On the other hand,
at this time, any other valid bin can contain a total number of
items of at most $\bf \Delta$.


These properties are true unless $S'=\emptyset$. In that case,
only large items remain to be packed, so the number of items in
any additional bin is at most $\frac 1{\eps}$. In this case we let
${\bf \Delta} =\frac{1}{\eps}$.


{\bf Approximating the cost function $f$.} Given the function $f$
we compute a staircase function, which is an
$(1+\eps)$-approximation of $f$, with $O(\log_{1+\eps} f(n))$
breakpoints. That is, we find a sequence of integers $0=k_0 <
k_1=1< \cdots < k_{\frac{1}{\eps}}=\frac 1{\eps}<
k_{\frac{1}{\eps} +1}< \cdots <k_{\ell}=n$ such that for all
$i=\frac 1{\eps},\frac 1{\eps}+1,\ldots ,\ell-1$, we have
$f(k_{i+1}) \leq (1+\eps) f(k_{i})$. The sequence is constructed
as follows. We define $k_j=j$ for $j=0,1,\ldots,\frac 1{\eps}$.
Every subsequent value $k_{j+1}$ for $j \geq \frac{1}{\eps}$ is
defined as the maximum integer $t>k_j$ such that $f(t) \leq
(1+\eps)f(k_j)$. Note that this definition is valid since for $j
\geq {1\over \eps}$ we have $f(j+1) \leq f((1+\eps)j) \leq
(1+\eps)f(j)$, where the first inequality holds by the
monotonicity of $f$, and the second inequality holds by the
concavity of $f$. Then, by the definition of the sequence, for
every $i=\frac 1{\eps},\frac 1{\eps}+1,\ldots ,\ell-2$, we have
$f(k_{i+2}) > (1+\eps) f(k_{i})$. Note that by the definition of
this sequence, we have $\ell =O({1\over \eps}+\log_{1+\eps} f(n))$
and $\ell \leq n$. Let $p_{\Delta}$ be such that $k_{p_{\Delta}}
\geq \bf \Delta$ and $k_{p_{\Delta}-1} \leq \bf \Delta$. If
$S'=\emptyset$, we have $\bf \Delta=\frac{1}{\eps}$, so
$k_{p_{\Delta}}=\frac{1}{\eps}$.  The staircase function, which is
an $(1+\eps)$-approximation of $f$, is defined as the value of $f$
for values $k_i$, and it remains constant between these points.

{\bf Constructing the linear program.} Given the instance $I''$,
we let a configuration of a bin $C$ be a (possibly empty) set of
items of $L'$ whose total (rounded-up) size is at most 1. We
denote the set of all configurations by $\mathring{\C}$. For each
configuration $C$ we define $p_{\Delta}+1 \leq \ell +1$ {\it
extended configurations} $(C,k_0),(C,k_1),\ldots
(C,k_{p_{\Delta}})$. A bin packed according to an extended
configuration $(C,k_{p})$ has large items according to
configuration $C$, and at most $k_p$ items in total (that are
either large or small items, i.e., including the large items of
this configuration). We later slightly relax this condition and
allow to increase the number of items in a bin (in favor of
possibly packing a slightly larger number of small items) in a way
that the cost of this bin only increases by a factor of $1+\eps$.
We denote by $\cal C$ the set of all extended feasible
configurations, where an extended configuration $(C,k_p)$ is
infeasible if the number of large items in $C$ is strictly above
$k_p$, and otherwise it is feasible. Let $H$ be the set of
different rounded-up sizes of large items. For each $v\in H$ we
denote by $n(v,C)$ the number of items with size $v$ in
configuration $C$, and we denote by $n(v)$ the number of items in
$L'$ with size $v$.

We denote the minimum size of an item by $s_{min}= \min_{i\in S'}
s'_i$ (note that $s_{min}\neq 0$), and we let
$s'_{min}=\max\{\frac{1}{(1+\eps)^t}|t\in \mathbb{Z}, \
\frac{1}{(1+\eps)^t}\leq s_{min}\}$ to be an approximated value of
$s_{min}$ which is an integer power of $1+\eps$. The value
$\log_{1+\eps}{\frac{1}{s'_{min}}}$ is polynomial in the size of
the input and in ${1\over \eps}$. We define the following set
$\W=\{(\frac{1}{(1+\eps)^t},k_{a})|0\leq t \leq
\log_{1+\eps}{\frac{1}{s'_{min}}}+1, 0\leq a\leq \ell\}$. A {\it
window} is defined as a member of $\W$.  The intuitive meaning of
a window here is a pair consisting of a bound on the remaining
capacity for small items in a bin (this bound is rounded to an
integer power of $1+\eps$), and a bound on the number of small
items packed into a bin.  $\W$ is also called the set of all
possible windows. Then, $|{\cal W}| \leq (\ell+1)\cdot
(\log_{1+\eps} {1\over s'_{min}}+2) $. For two windows, $w^1$ and
$w^2$ where $w^i=(w^i_s,w^i_n)$ for $i=1,2$, we say that $w^1 \leq
w^2$ if $w^1_n\leq w^2_n$ and $w^1_s \leq w^2_s$.

Note that each bin that contains large items, packed according to
an extended configuration $(C,k_p)$, may leave space for small
items.  For an extended configuration $(C,k_p)$ we denote the {\it
main window of $(C,k_p)$} to be $w(C,k_p)=(w(C),n(C,k_p))$, where
$w(C)$ is an approximation of the available size for small items
in a bin with configuration $C$, and $n(C,k_p)$ is an upper bound
on the total number of small items that can fit into this bin.
More precisely, assume that the total (rounded-up) size of the
items in $C$ is $s'(C)$. We let $w(C)={1\over (1+\eps)^t}$ where
$t$ is the maximum integer such that $ 0\leq t \leq
\log_{1+\eps}{\frac{1}{s'_{min}}}+1$ and that $s'(C) + {1\over
(1+\eps)^t} \geq 1$.

\begin{corollary}
\label{realcost} Given an extended configuration $(C,k_p)$, the
real cost (after adding small items such that their number is not
larger than the number in the main window of $(C,k_p)$) of a bin
that is packed according to this extended configuration, is at
most $(1+\eps)f(k_p)$.
\end{corollary}
\begin{proof}
Assume that in configuration $C$ we pack $n_C=\sum\limits_{v \in
H}n(v,c)$ large items, then let $t$ be the smallest integer such
that $ k_p-n_C\leq k_t$. It can be seen that $t \leq p$ always
holds. We let $n(C,k_p)=k_t$.  Note that if $ k_p-n_C\neq k_t$
then $k_t > {1\over \eps}$, and $t > \frac{1}{\eps}$, so we have
$k_p -n_C
> k_{t-1}$.
Hence in this case we conclude that $f(n(C,k_p)+n_C) = f(k_t+n_C)
\leq f(k_t+k_p-k_{t-1}) \leq f(k_p) + f(k_t)- f(k_{t-1}) \leq
f(k_p)+\eps f(k_{t-1}) \leq (1+\eps)f(k_p)$, where the first
inequality holds by the definition of $t$ and the monotonicity of
$f$, the second inequality holds by the concavity of $f$ (since
$k_t>k_{t-1}$), the third inequality holds because $f(k_t) \leq
(1+\eps)f(k_{t-1})$ and the last inequality holds by the
monotonicity of $f$ (since $k_{t-1}<k_p-n_C\leq k_p$). Moreover,
if $ k_p-n_C= k_t$, then $f(n(C,k_p)+n_C)\leq (1+\eps)f(k_p)$
clearly holds as well.
\end{proof}

The main window of an extended configuration is a window (i.e., it
belongs to $\W$), but $\W$ may include windows that are not the
main window of any extended configuration. We note that $|{\cal
W}|$ is polynomial in the input size and in $1\over \eps$, whereas
$|{\cal C}|$ may be exponential in $1\over \eps$ (specifically,
$|{\cal C}| \leq \ell\cdot({1\over \eps^3}+1)^{1/\eps}$). We
denote the set of windows that are actual main windows of at least
one extended configuration by ${\cal {W}'}$. We first define a
linear program that allows the usage of any window in $\W$. After
we obtain a solution to this linear program, we modify it so that
it only uses windows of ${\cal W}'$.

We define a generalized configuration $\tilde{C}$ as a pair of
pairs $\tilde{C}=((C,k_p),W=(w,k_j))$, for some feasible extended
configuration $(C,k_p)$ and some $W\in \W$. The generalized
configuration $\tilde{C}$ is valid if $W\leq w(C,k_p)$. The set of
all valid generalized configurations is denoted by $\tilde{\C}$.

For $W\in {\cal W}$ denote by $C(W)$ the set of valid generalized
configurations $\tilde{C}=((C,k_p),W')$ such that $W$ is their
window, i.e., $C(W)=\{ ((C,k_p),W')\in \tilde{\C} : W'=W \}$.

We next consider the following linear program.   In this linear
program we have a variable $x_{\tilde{C}}$ denoting the number of
bins with generalized configuration $\tilde{C}$, and variables
$Y_{i,W}$ indicating if the small item $i$ is packed in a window
of type $W$ (the exact instance of this window is not specified in
a solution of the linear program).

\begin{eqnarray}
\min & \sum\limits_{{\tilde{C}=((C,k_p),W)}\in {\tilde{\C}}} f(k_p) x_{\tilde{C}}& \nonumber \\
s.t.& \sum\limits_{\tilde{C}=((C,k_p),W)\in {\tilde \C}} n(v,C) x_{\tilde {C}} \geq n(v) & \forall v\in H\label{gcc1}\\
& \sum\limits_{W\in {\cal W}} Y_{i,W} \geq 1& \forall i\in S'\label{gcc2}\\
&  w\cdot \sum\limits_{\tilde{C} \in C(W)} x_{{\tilde{C}}} \geq
\sum\limits_{i\in S'} s'_i\cdot Y_{i,W}
& \forall W=(w,\kappa)\in {\cal W}\label{gcc3}\\
& \kappa\cdot \sum\limits_{{\tilde{C} \in C(W)}} x_{\tilde{C}}
\geq \sum\limits_{i\in S'} Y_{i,W} & \forall W=(w,\kappa)\in {\cal
W}
\label{gcc4}\\
& x_{\tilde{C}} \geq 0& \forall \tilde{C}\in {\tilde{\C}}\nonumber
\\
& Y_{i,W} \geq 0&  \forall W\in {\cal W}, \forall i \in
S'.\nonumber
\end{eqnarray}
Constraints (\ref{gcc1}) and (\ref{gcc2}) ensure that each item
(large or small) of $I''$ will be considered. The large items will
be packed by the solution, and the small items would be assigned
to some type of window. Constraints (\ref{gcc3}) ensure that the
total size of the small items that we decide to pack in window of
type $W$ is not larger than the total available size in all the
bins that are packed according to a generalized configuration,
whose window is of type $W$ (according to the window size).
Similarly, the family of constraints (\ref{gcc4}) ensures that the
total number of the small items that we decide to pack in a window
of type $W$ is not larger than the total number of small items
that can be packed (in accord with the second component of $W$) in
all the bins whose generalized configuration of large items
induces a window of type $W$.  In the sequel we show how to deal
with small items and specifically, how to pack most of them into
the windows allocated for them, and how to further deal with some
unpacked small items.

\begin{lemma} There is a feasible solution to the above linear program that
has a cost of at most $(1+\eps) \opt(I'')$.
\end{lemma}
\begin{proof}
The $(1+\eps)$ factor results from the fact that we define
extended configurations, where the number of items per bin is
$k_p$ (for some value of $p$). The fact that we use a window
$(w,\kappa)$ only for values of $\kappa$ that belong to the same
sequence of values $k_i$ will result in an additional factor of
$1+\eps$ on the cost of the linear program.

To convert the solution, we do not need to modify packing of
items, but we change the cost calculation of each bin to comply
with costs of generalized configurations. For this, the number of
items in every bin must be converted (in favor of cost
calculations) as follows.

Given a bin with $n_1 > 0 $ items, we define $p$ to be minimal
value such that $k_p \geq n_1$. The increase in the cost can occur
if $k_p > n_1$. In this case, $p>0$ and we have $k_{p-1} <n_1<k_p$
and thus using monotonicity of $f$ and the properties of the
sequence $k_i$ we have $f(k_{p})\leq (1+\eps)f(k_{p-1})\leq
(1+\eps)f(n_1)$.  Since windows are never smaller than the real
space in bins, both with respect to size and with respect to the
difference between the number of large items and the value $k_p$
of the configuration, the solution clearly satisfies the
constraints (\ref{gcc3}) and (\ref{gcc4}) on the packing of small
items, and the packing of large items satisfies the constraints
(\ref{gcc1}). Therefore the adapted solution is a feasible
solution of the linear program. Moreover, the adapted solution
implies a solution to the linear program in which all variables
$x_{\tilde{C}}$, that correspond to generalized configurations
$\tilde{C}=(C,w)$ for which $w$ is not the main window of $C$, are
equal to zero, and all variables $Y_{i,w}$ where $w \notin \W'$
are equal to zero as well. The linear program calculates the cost
of a packing using the values $k_p$ of the extended
configurations, and as shown above, this increases the cost of
$\opt(I'')$ by a multiplicative factor of at most $1+\eps$ (see
Corollary \ref{realcost}).
\end{proof}

{\bf The column generation technique.} We invoke the column
generation technique of Karmarkar and Karp \cite{KK82} as follows.
The above linear program may have an exponential number of
variables and polynomial number of constraints (neglecting the
non-negativity constraints). Instead of solving the linear program
we solve its dual program (that has a polynomial number of
variables and an exponential number of constraints) that we
describe next.

The dual variables $\alpha_v$ correspond to the item sizes in $H$,
and the dual variables $\beta_i$ correspond to the small items of
$S'$. The intuitive meaning of these two types of variables can be
seen as weights of these items. For each $W\in {\cal W}$ we have a
pair of dual variables $\gamma_W,\delta_W$.  Using these dual
variables, the dual linear program is as follows.

\begin{eqnarray}
\max & \sum\limits_{v \in H} n(v) \alpha_v+\sum\limits_{i\in S'} \beta_i&\nonumber \\
s.t. & \sum\limits_{v\in H} n(v,C) \alpha_v +
w\gamma_{W}+\kappa\delta_W \leq f(k_p)& \forall
\tilde{C}=((C,k_p),W=(w,\kappa))\in {\tilde{\C}}\
\label{gcd1}\\
& \beta_i - s'_i\gamma_W -\delta_W \leq 0&\forall i\in S', \forall W\in {\cal W}\label{gcd2}\\
 & \alpha_v \geq 0& \forall v\in H\nonumber \\&
\beta_i\geq 0& \forall i\in S'\nonumber \\ & \gamma_W,\delta_W
\geq 0& \forall W\in {\cal W}.\nonumber
\end{eqnarray}

First note that there is a polynomial number of constraints of
type (\ref{gcd2}), and therefore we clearly have a polynomial time
separation oracle for these constraints.
 If we would like to solve
the above dual linear program (exactly) then using the ellipsoid
method we need to establish the existence of a polynomial time
separation oracle for the constraints (\ref{gcd1}).
 However, we
are willing to settle on an approximated solution to this dual
program.
 To be able to apply the ellipsoid algorithm, in order to
solve the above dual problem within a factor of $1+\eps$, it
suffices to show that there exists a polynomial time algorithm
(polynomial in $n$, $1\over \eps$ and $\log \frac 1{s'_{min}}$ and
$\log f(n)$) such that for a given solution
$a^*=(\alpha^*,\beta^*,\gamma^*,\delta^*)$ decides whether $a^*$
is a feasible dual solution (approximately). That is, it either
provides a generalized configuration
$\tilde{C}=((C,k_p),W=(w,k_t)) \in \tilde{\C}$ for which
$\sum\limits_{v\in H} n(v,C) \alpha^*_v + w\gamma^*_W+k_t
\delta^*_W> 1$, or outputs that an approximate infeasibility
evidence does not exist, that is, for all generalized
configurations $\tilde{C}=((C,k_p),W=(w,k_t)) \in \tilde{\C}$,
$\sum\limits_{v\in H} n(v,C) \alpha^*_v + w\gamma^*_W+k_t
\delta^*_W \leq 1+\eps$ holds. In such a case, $a^*\over 1+\eps$
is a feasible dual solution which also satisfies constraints
(\ref{gcd2}), that can be used.

Our algorithm for finding an approximate infeasibility evidence
uses the following problem as an auxiliary problem. The {\sc
knapsack problem with a maximum cardinality constraint} (KCC)
problem is defined as follows. Given a set of item types $H$ and
an integer value $k$, where each item type $v \in H$ has a given
multiplicity $n(v)$, a volume $z^*_v$ and a size $v$, the goal is
to pack a multiset of at most $k$ items (taking the multiplicity,
in which items are taken, into account, and letting the solution
contain at most $n(v)$ items of type $v$) and a total size of at
most 1, so that the total volume is maximized.  To provide an
FPTAS for KCC, note that one can replace an item with size $v$ by
$n(v)$ copies of this item and then one can apply the FPTAS of
Caprara et al. \cite{CKPP} for the knapsack problem with
cardinality constraints. The FPTAS of \cite{CKPP} clearly has
polynomial time in the size of its input, and $\frac 1{\eps}$.
Since the number of items that we give to this algorithm as input
is at most $n$, we can use this FPTAS and still let our scheme
have polynomial running time.

A configuration $\tilde{C}$, that is an approximate infeasibility
evidence, can be found by the following procedure:  For each
$W=(w,k_t) \in {\cal W}$, and for every $0\leq p \leq \ell$, we
look for an extended configuration $(C,k_p) \in \C$ such that
$((C,k_p),W)$ is a valid generalized configuration, and such that
$\sum\limits_{v\in H} n(v,C) \alpha^*_v$ is maximized. If a
configuration $C$ is indeed found, the generalized configuration,
whose constraint is checked, is $((C,k_p),W)$. To find $C$, we
invoke the FPTAS for the  KCC problem with the following input:
The set of items is $H$ where for each $v\in H$ there is a volume
$\alpha^*_v$ and a size $v$, the goal is to pack a multiset of the
items,  so that the total volume is maximized, under the following
conditions. The multiset should consist of at most $k_p-k_{t-1}-1$
large items, (taking the multiplicity into account, but an item
can appear at most a given number of times). If $t=0$, we instead
search for a multiset with at most $k_p$ large items. The total
(rounded-up) size of the multiset should be smaller than
$1-\frac{w}{1+\eps}$, unless $w<s'_{min}$, where the total size
should be at most 1 (in this case, the window does not leave space
for small items). Since the number of applications of the FPTAS
for the KCC problem is polynomial (i.e., $(\ell+1)|{\cal W}|$),
this algorithm runs in polynomial time.

If it finds a solution, that is, a configuration $C$, with at most
$k_p-k_{t-1}-1$ large items (or $k_p$, if $t=0$), and a total
volume greater than $f(k_p)-w\gamma^*_W-\kappa\delta^*_W$, we
argue that $((C,k_p),(w,k_t))$ is indeed a valid generalized
configuration, and this implies that there exists a generalized
configuration, whose dual constraint (\ref{gcd1}) is violated.
First, we need to show that $(C,k_p)$ is a valid extended
configuration. This holds since $C$ has at most $k_p-k_{t-1}-1\leq
k_p$ large items (if $t=0$ the bound on the number of items holds
immediately).

By the definition of windows, the property $w<s'_{min}$ is
equivalent to $w=\frac{s'_{min}}{1+\eps}$, which is the smallest
size of window (and the smallest sized window forms a valid
generalized configuration with any configuration, provided that
the value of $k_t$ is small enough). If $t>0$, since $C$ has at
most $k_p-k_{t-1}-1$ items, the second component of the main
window of $C$ in this case is larger than $k_{t-1}$ and thus no
smaller than $k_t$, and the window is no smaller than $(w,k_t)$.
Therefore, the generalized configuration $((C,k_p),(w,k_t))$ is
valid. If $t=0$ then the window $(w,0)$ is clearly valid with any
extended configuration (for the current value of $w$).

If $w \geq s'_{min}$, recall that the main window of $(C,k_p)$,
$w(C,k_p)=(w(C),n(C,k_p))$ is chosen so that $s'(C)+{w(C)} \geq
1$, and that $C$ is chosen by the algorithm for KCC so that
$s'(C)<1-\frac{w}{1+\eps}$. We get $1-w(C) \leq
s'(C)<1-\frac{w}{1+\eps}$ and therefore $w < (1+\eps)w(C)$, i.e.,
$w \leq w(C)$ (since the sizes of windows are integer powers of
$1+\eps$). Since $C$ contains at most $k_p-k_{t-1}-1$ items, we
have  $n(C,k_p) \geq k_t$ and so we conclude that $W \leq
w(C,k_p)$, and $((C,k_p),W)$ is  a valid generalized configuration
(the same property holds for $t=0$). Thus in this case we found
that this solution is a configuration whose constraint in the dual
linear program is not satisfied, and we can continue with the
application of the ellipsoid algorithm.

Otherwise, for any pair of a  window $W=(w,k_t)$, and a value
$0\leq p\leq \ell$, and any configuration $C$ of total rounded-up
size less than $1-\frac{w}{1+\eps}$ (or at most 1, if
$w<s'_{min}$), with at most $k_p-k_{t-1}-1$ items, has a volume of
at most $(1+\eps)(1-w\gamma^*_W-k_t\delta^*_W)\leq
(1+\eps)-w\gamma^*_W-k_t \delta^*_W $. We prove that in this case,
all the constraints of the dual linear program are satisfied by
the solution $a^*\over 1+\eps$. Consider an arbitrary valid
generalized configuration $\tilde{C}=((C,k_p),(\tilde{w},k_j))$,
where $(C,k_p)$ is a valid extended configuration. We have
$(\tilde{w},k_j)\leq (w(C),n(C,k_p))$, where $(w(C),n(C,k_p))$ is
the main window of $C$. If $w(C)<s'_{min}$, then $\tilde{w}=w(C)$.
Since $s'(C)\leq 1$ for any configuration, and $k_j \leq
n(C,k_p)$, we prove that the number of items in $C$ is at most
$k_p-k_{j-1}-1$ (if $j=0$ then the number of items in $C$ is
immediately at most $k_p$ and there is nothing to prove). Assume
by contradiction that the number of items in $C$ is at least
$k_p-k_{j-1}$. Then by definition, we have $n(C,k_p) \leq
k_{j-1}$, which is impossible. Thus, $(C,k_p)$ is a possible
extended configuration to be used with the window
$(\tilde{w},k_j)$ in the application of the FPTAS for KCC, or $C$
is a possible configuration to be used with the parameter $p$ and
the window $(\tilde{w},k_j)$ in the application of the FPTAS for
KCC. Assume next that $\tilde{w} <1$, then when the FPTAS for KCC
is applied on $W=(\tilde{w},k_j)$, $C$ is a configuration that is
taken into account for $W$ since $s'(C)<1-\frac{w(C)}{1+\eps}\leq
1-\frac{\tilde{w}}{1+\eps}$, where the first inequality holds by
definition of $w_s(C)$, and $C$ has at most $k_p-k_{j-1}-1$ items.
If $\tilde{w} =1$ then $1 \geq w(C) \geq \tilde{w}=1$, so
$w_s(C)=1$. A configuration $C_1$ that contains at least one large
item satisfies $s'(C_1)\geq \eps$, so $s'(C_1)+\frac{1}{1+\eps}
\geq \frac{1+\eps+\eps^2}{1+\eps}>1$. Therefore if the main window
of a configuration is of size 1, this configuration is empty. We
therefore have that $C$ is an empty configuration, thus $s'(C)=0$.
The extended configuration $(C,k_p)$ is valid for any $0\leq p
\leq \ell$. We have $n(C,k_p)=k_p$ for the empty configuration,
and for any $1 \leq j \leq p$, $k_p-k_{j-1}-1\geq 0$, and for
$j=0$, $k_p\geq 0$. This empty configuration is considered with
any window $W=(w,k_j)\in \W$ where $j>0$ in the application of
KCC. Note that if $j=0$, the configuration has no items at all
(large or small).

We denote by $(X^*,Y^*)$ the solution to the primal linear program
that we obtained.

\begin{lemma}\label{lg1}
The cost of $(X^*,Y^*)$ is at most $(1+\eps)^2\opt(I'')$.
\end{lemma}
\begin{proof}
The solution $(X^*,Y^*)$  is a $(1+\eps)$ approximation for the
optimal solution to the linear program. Since we showed that there
exists a feasible solution to the primal linear program with a
cost of at most $(1+\eps)\opt(I'')$, we conclude that
$\sum\limits_{\tilde{C}=((C,k_p),(w,k_t))\in {\tilde{\cal
C}}}f(k_p)X^*_{\tilde{C}} \leq (1+\eps)^2\opt(I'')$.
\end{proof}

{\bf Modifying the solution to the linear program so that all
windows in $\cal{W}\setminus \cal{W}'$ can be neglected.} We
modify the solution to the primal linear program, into a different
feasible solution of the linear program, without increasing the
goal function. We create a list of generalized configurations
whose $X^*$ component is positive. From this list of generalized
configurations, we find a list of windows that are the main window
of at least one extended configuration induced by a generalized
configuration in the list. This list of windows is a subset of
$\cal{W}'$ defined above. We would like the solution to use only
windows from $\cal{W}'$.

The new solution will have the property that any non-zero
components of $X^*$, $X^*_{\tilde{C}}$ corresponds to a
generalized configuration $\tilde{C}=((C,k_p),W)$, such that $W\in
\cal{W}'$. We still allow generalized configurations
$\tilde{C}=((C,k_p),W)$ where $W$ is not the main window of
$(C,k_p)$, as long as $W\in \W'$. This is done in the following
way. Given a window $W' \notin \W'$, we define
$B_{W'}=\sum\limits_{\tilde{C''}\in C(W')}x^*_{\tilde{C''}}$. The
following is done in parallel for every generalized configuration
$\tilde{C'}=((C,k_p),W')$, where $W'\notin \W'$ and such that
$X^*_{\tilde{C'}}>0$, where the main window of $(C,k_p)$ is $W
\geq W'$ (but $W'\neq W$). We let $\tilde{C}=((C,k_p),W)$. The
windows allocated for small items need to be modified first, thus
an amount of $\frac{X^*_{\tilde{C'}}}{B_{W'}}Y^*_{i,W'}$ is
transferred from $Y^*_{i,W'}$ to $Y^*_{i,W}$. We modify the values
$X^*_{\tilde{C'}}$ and $X^*_{\tilde{C}}$ as follows. We increase
the value of $X^*_{\tilde{C}}$ by an additive factor of
$X^*_{\tilde{C'}}$ and let $X^*_{\tilde{C'}}=0$.

To show that the new vector $(X^*,Y^*)$ still gives a feasible
solution of the same value of objective function, we consider the
modifications. For every extended configuration $(C,k_p)$, the sum
of components $X^*$, that correspond to  generalized
configurations whose extended configuration of large items is
$(C,k_p)$, does not change. Therefore, the value of the objective
function is the same, and the constraints (\ref{gcc1}) still hold.
We next consider the constraint (\ref{gcc2}) for $i$, for a given
small item $i\in S'$. Since the sum of variables $Y^*_{i,W}$ does
not change, this constraint still holds.

As for constraints (\ref{gcc3}) and (\ref{gcc4}), for a window $W'
\notin \cal{W'}$, the right hand side of each such constraint
became zero. On the other hand, for windows in $\cal{W'}$, every
increase in some variable $X^*_{\tilde{C}}$ for
$\tilde{C}=((C,k_p),W=(w,\kappa))$, that is originated in a
decrease of $X^*_{\tilde{C'}}$ for
$\tilde{C}=((C,k_p),W'=(w',\kappa'))$ is accompanied with an
increase of $\frac{X^*_{\tilde{C'}}}{\sum\limits_{\tilde{C''}\in
C(W')}X^*_{\tilde{C''}}}Y^*_{i,W'}=\frac{X^*_{\tilde{C'}}}{B_{W'}}Y^*_{i,W'}
$ in $Y^*_{i,W}$, for every $i\in S'$, thus is, an increase of
$\sum\limits_{i\in S'}\frac{X^*_{\tilde{C'}}}{B_{W'}} s'_i\cdot
Y^*_{i,W'}$ in the right hand size of the constraint (\ref{gcc3})
for $W$, and an increase of $w \cdot X^*_{\tilde{C'}}$ in the left
hand side. Since we have $w \cdot B_{W'} \geq w' \cdot B_{W'} \geq
\sum\limits_{i\in S'} s'_i\cdot Y^*_{i,W'}$ before the
modification occurs (since constraint (\ref{gcc3}) holds for the
solution before modification for the window $W'$), we get that the
increase of the left hand side is no smaller than the increase in
the right hand side. There is an increase of $\sum\limits_{i\in
S'}\frac{X^*_{\tilde{C'}}}{B_{W'}} \cdot Y^*_{i,W'}$ in the right
hand size of the constraint (\ref{gcc4}) for $W$, and an increase
of $\kappa \cdot X^*_{\tilde{C'}}$ in the left hand side. Since we
have $\kappa \cdot B_{W'} \geq \kappa' \cdot B_{W'} \geq
\sum\limits_{i\in S'} Y^*_{i,W'}$, we get that the increase of the
left hand side is no smaller than the increase in the right hand
side.

Now, we can temporarily delete the constraints of (\ref{gcc3}) and
(\ref{gcc4}) that correspond to windows in ${\cal W} \setminus
{\cal W}'$. We call the resulting linear program $LP_{tmp}$. We
consider a basic solution of $LP_{tmp}$ that is not worse than the
solution we obtained above (which was created as a solution of
$LP_{tmp}$ too). Such a basic solution can be found in polynomial
time.  We denote this basic solution by $(\X^*,\Y^*)$. This is
clearly a basic solution to the original linear program as well.

In order to obtain a feasible packing, we need to use the solution
$(\X^*,\Y^*)$. However, this solution may contain fractional
components. We can show the following bound on these components.

\begin{lemma}
\label{frac_comp} Consider the solution $(\X^*,\Y^*)$. Let $F_Y$
be the number of small items that are assigned to windows
fractionally according to the solution, i.e., $F_Y=|\{i\in S',$
such that the vector $(\Y^*_{i,W})_{W\in {\cal W}}$ is
fractional$\}|$. Let $F_X$ be the number of fractional components
of $\X^*$, i.e., the number of configurations assigned a
non-integer number of copies in the solution. Then $F_Y+F_X \leq
|H|+2|{\cal W'}|$.
\end{lemma}
\begin{proof}
The linear program $LP_{tmp}$ consists of $|H|+2|{\cal W'}|+|S'|$
inequality constraints, and hence in a basic solution (a property
that we assume that  $(\X^*,\Y^*)$ satisfies) there are at most
$|H|+2|{\cal W'}|+|S'|$ basic variables. For every $i\in S'$,
there is at least one window $W$ such that $Y_{i,W}$ is a basic
variable, and therefore there are at most $|H|+2|{\cal W'}|$
additional fractional components in $(\X^*,\Y^*)$. \end{proof}

{\bf Rounding the solution.}  We apply several steps of rounding
to obtain a feasible packing of the items into bins. Let $C_{LP}$
be the cost obtained in the linear program by the vector
$(\X^*,\Y^*)$. By Lemma \ref{lg1}, this cost is at most
$(1+\eps)^2 \opt(I'')$.

For each $i\in S'$ such that the vector $(\Y^*_{i,W})_{W\in {\cal
W}}$ is fractional, $i$ is packed in a dedicated bin. We can
therefore assume that for every small item $i \in S'$ to be
packed, $(\Y^*_{i,W})_{W\in {\cal W}}$ is integral. Without loss
of generality, we assume that it has one component equal to 1, and
all other components are zero. (If this is not the case, we can
modify the vector without changing the feasibility of the
solution, or the value of the objective function.)

Let $\hat{X}$ be the vector such that $\hat{X}_{\tilde{C}}=\lceil
\X^*_{\tilde{C}} \rceil$ for all $\tilde{C} \in {\tilde{\C}}$. The
number of bins allocated to generalized configuration $\tilde{C}$
is $\hat{X}_{\tilde{C}}$.

We pack the items of $L'$ first. We initialize bins according to
generalized configurations, and assign large items into these bins
according to the associated configurations (some slots may remain
empty).

\begin{lemma}\label{lg2}
The cost of the additional bins, dedicated to small items for
which $(\Y^*_{i,W})_{W\in {\cal W}}$ is fractional, and the cost
of additional bins that are created as a result of replacing
$\X^*$ by $\hat{X}$ is at most $f(k_{p_{\Delta}})\cdot
(|H|+2|{\cal W'}|)$.
\end{lemma}
\begin{proof}
We calculate the cost of bins opened in addition to the cost
implied by the solution $(\X^*,\Y^*)$. At most one bin containing
at most $k_{p_{\Delta}}$ items was opened for every fractional
component of $\X^*_{\tilde{C}}$. At most one bin containing a
single item was opened for every small item that was assigned
fractionally to windows. The cost of a bin of the first type is at
most $f(k_{p_{\Delta}})$. The cost of every bin of the second type
is  $f(1)=1\leq f(\frac{1}{\eps})\leq f(k_{p_{\Delta}})$. The
total number of the two types of bins together is at most
$|H|+2|{\cal W'}|$ by Lemma \ref{frac_comp}.
\end{proof}

Before moving on to the specific assignment of small items, we
complete the packing of the original large items. Each large item
of the rounded-up instance is replaced by the corresponding item
of $I$. The method of rounding implies that the space allocated to
the rounded items is sufficient for the original items. Moreover,
every item is replaced by at most one item, so the cost does not
increase.

Each item of $L_1$ is packed into one dedicated bin.
\begin{lemma}\label{lg3}
The cost of the bins dedicated to the items of $L_1$ is at most
$2\eps^2 \opt(I'')$.
\end{lemma}
\begin{proof}
It suffices to show that $f(1)|L_1| \leq 2\eps^2\opt(I'')$.  To
see this last claim note that $|L_1|\leq 2|L|\eps^3$ and each item
in $L$ has size at least $\eps$ and therefore the number of bins
used by $\opt(I'')$ is at least $|L|\eps$, where each of them
costs at least $f(1)$. Therefore,  $f(1) |L_1| =|L_1|\leq 2\eps^2
\opt(I'')$.
\end{proof}


By the constraints (\ref{gcc1}), the allocation of the items of
$L'$ to slots reserved for such items is successful. At this time,
we have removed some small items into new bins, and possibly
increases the space allocated to other small items.

%

We next consider the packing of the small items that are supposed
to be packed (according to $\Y^*$) in bins with window $W$. Assume
that there are $X(W)$ such bins (i.e.,
$X(W)=\sum\limits_{\tilde{C}=((C,k_t),W)} \hat{X}_{\tilde{C}}$).
Denote by $S(W)$ the set of small items of $S'$ that we decided to
pack in bins with window $W$ (for some of these items we will
change this decision in the sequel). Then, by the feasibility of
the linear program we conclude that $\sum\limits_{i\in S(W)} s'_i
\leq w\cdot X(W)$ and $|S(W)| \leq k_p\cdot X(W)$ for any
$W=(w,k_p)\in {\cal W'}$.

We next show how to allocate almost all the items of $S(W)$ to the
$X(W)$ bins with window $W=(w,k_p)$ such that the total size of
items of $S(W)$ in each such bin will be at most $1+\frac{\eps
w}{1+\eps}$ and the total number of items of $S(W)$ in each such
bin will be at most $k_p$.

To do so, we sort the items in $S(W)$ according to non-increasing
size  (assume the sorted list of item indices is $b_1<b_2< \ldots
b_{|S(W)|}$). Then, allocate the items to the bins in a
round-robin manner, so that bin $j$ ($1 \leq j \leq S(W)$)
receives items of indices $b_{j+q\cdot X(W)}$ for all integers
$q\geq 0$ such that $j+q\cdot X(W) \leq |S(W)|$.  We call the
allocation of items for a given value of $p$ a {\it round of
allocations}. If $w=\frac{s'_{min}}{1+\eps}$ then there are no
small items assigned to this window.  We therefore assume $w \geq
s'_{min}$.

We claim that the last bin of index $X(W)$ received at most an
$\frac{1}{X(W)}$ fraction of the total size of the items, whose
sum is equal to $\sum\limits_{i=1}^{|S(W)|} s_{b_i}$. To prove
this, we artificially add at most $X(W)-1$ items of size zero to
the end of the list (these items are added just for the sake of
the proof), and allocate them to the bins that previously did not
receive an item in the last round of allocations, that is, bins
$r,\ldots,X(W)$ such that bin $r-1<X(W)$ originally received the
last item. If bin $X(W)$ received the last item then no items are
added. Now the total size of small items remained the same, but
every bin got exactly one item in each round. Since the last bin
received the smallest item in each round, the claim follows. On
the other hand, we can apply the following process, at every time
$i < X(W)$, remove the first (largest) small item from bin $i$. As
a result, the round-robin assignment now starts from bin $i+1$ and
bin $i$ becomes the bin that receives items last in every round,
and thus by the previous proof, the total size of items assigned
to it is at most $\frac{\sum\limits_{i=1}^{|S(W)|} s_{b_i}}{X(W)}$
(since the total size of items does not increase in each step of
removal).

We create an intermediate solution $SOL_{inter}$ by removing the
largest small item from each such bin (call them {\it the removed
small items}). Each removed item is small and therefore its size
is at most $\eps$. We pack the removed small items in new bins, so
that each bin contains $\frac 1{\eps}$ items. There may be at most
one resulting bin with less than $\frac 1{\eps}$ items.

The solution $SOL_{inter}$ is not necessarily valid, but if we
temporarily relax the condition on the total size of items in a
bin, we can compute its cost. Since the assignment of small items
into bins is done using a round-robin method, the number of small
items in a bin with a window $(w,k_p)$ is at most $k_p$.

\begin{lemma}\label{linterg}
The total cost of $SOL_{inter}$ is at most the sum of $f({1\over
\eps})$ plus $(1+\eps)^2$ times the cost of the solution prior to
the allocation of the small items into bins.
\end{lemma}
\begin{proof}
 The first
factor of $1+\eps$ follows from Corollary \ref{realcost}. We
calculate the cost of the additional bins. We allocate a cost of
$\eps f({1\over \eps})$ to each removed small item. Then, the
total allocated cost covers that cost of all new bins except for
at most one bin that has a cost of at most $f(\frac 1{\eps})$.
Consider a removed item $i$ and let $a$ be the real number of
items (including large items) that the bin from which $i$ is
removed, contains before the removal. Thus, the bin is charged
with a cost of at least $f(a)$ (the linear program may have
charged it with $f(k_p)$ for some $k_p \leq a$, but the current
charge for this bin in our estimation of the total cost is
$(1+\eps)f(k_p) \geq f(a)$, by Corollary \ref{realcost}). As a
result of removal of $i$, the real cost of the bin is no larger
than $f(a-1)$. We therefore show $\eps f(\frac{1}{\eps}) +f(a-1)
\leq (1+\eps)f(a)$. If $a \geq \frac{1}{\eps}$, then using
monotonicity, $f(\frac{1}{\eps}) \leq f(a)$ and $f(a-1) \leq f(a)$
so the claim holds. Otherwise, we have $f(\frac 1{\eps})=
f(a)+\sum\limits_{j=a+1}^{\frac{1}{\eps}} (f(j)-f(j-1))$. By
concavity, we have for every $a+1 \leq j \leq \frac{1}{\eps}$,
$f(a)-f(a-1) \geq f(j)-f(j-1)$. Therefore $f(\frac 1{\eps}) \leq
f(a)+\frac{1}{\eps}(f(a)-f(a-1))$. Rewriting this gives the
required claim.
\end{proof}

We note that the total size of small items assigned to such
(original) bin is at most $w$ (as before removing the items we
allocate the first bin a total size that is at most $w$ and after
the removal of items each bin has total size which is at most the
total size of the first bin before the removal).

The intermediate solution $SOL_{inter}$ is infeasible because our
definition of $w$ is larger than the available space for small
items in such bin. We create the final solution $SOL_{final}$ as
follows.

Consider a bin such that the intermediate solution in which large
items are packed according to configuration $C$, and small items
with total size at most $w$.  We do not change the packing of
large items. As for the small items, we remove them from the bin
and start packing the small items into this bin greedily in
non-decreasing order of the item sizes, as long as the total size
of items packed to the bin does not exceed 1. The first item that
does not fit into the bin is called the {\it special item}.
Additional items that do not fit are called the {\it excess
items}.

We collect the special items from all bins, and we pack these
items in separate bins, so that each such separate bin will
contain $1\over \eps$ special items for different bins of
$SOL_{inter}$, except for the last such bin. Similarly to the
above argument in the proof of Lemma \ref{linterg}, these are
feasible bins and they add an additive factor of $\eps$ times the
cost of $SOL_{inter}$ to the total cost of the packing (plus
$f(\frac 1{\eps})$).

By the definition of windows, the actual space in a bin with
window $(w,\kappa)$, that is free for the use of small items, is
at least of size $\frac{w}{1+\eps}$. After the removal of the
packed items and the special item, we are left with the excess
items, and their size is at most $w-\frac{w}{1+\eps}=\eps
\frac{w}{1+\eps} < \eps$.  Similar considerations can be applied
to the cardinality of these items. Since we insert the items into
the window sorted by a non-decreasing order of size, the largest
items are the ones that become excess items, and thus for a window
$(w,\kappa)$, the number of excess items is at most $\eps \kappa$.

The last rounding step is defined as follows. We can pack the
unpacked (excess) items of every $1\over \eps$ bins of
$SOL_{inter}$ using one additional bin. Specifically, we sort the
subsets of excess items according to a non-increasing order of the
second component of the windows to which these items were
originally assigned, we call it the {\it index of the subset}.
Then, according to this order, we assign every consecutive $1\over
\eps$ subsets to a bin. The last bin may contain a smaller number
of subsets. This completes the scheme. We get our final solution
$SOL_{final}$.

\begin{lemma}
The cost of $SOL_{final}$ is at most $(1+2\eps)$ times the cost of
$SOL_{inter}$ plus $f(k_{p_{\Delta}})$.
\end{lemma}
\begin{proof}
We use $\kappa_i$ to denote the index of the $i$-th subset. Let
$v$ denote the number of bins created, and $u$ the number of
subsets (we have $\frac{v-1}{\eps}<u\leq \frac{v}{\eps}$). The
number of items in the $i$-th bin, for $i \geq 2$, is at most $
\sum\limits_{j=1}^{\frac 1{\eps}} \eps \kappa_{\frac{i-1}{\eps}+j}
\leq \kappa_{\frac{i-1}{\eps}}$. The number of items in the first
bin is at most ${\bf \Delta} \leq k_{p_{\Delta}}$. The cost of the
bins is therefore at most $f({\bf \Delta})
+\sum\limits_{i=1}^{v-1} f\left(\kappa_{\frac{i}{\eps}}\right)$.
On the other hand, the cost of $SOL_{inter}$ that is charged to
the bins which was supposed to get the $i$-th subset of excess
items is at least $f(\kappa_i)$ (since for a generalized
configuration $((C,k_t),(w,k_p))$ we have $k_p \leq k_t$), thus
the cost of $SOL_{inter}$ is at least $\sum\limits_{j=1}^u
f(\kappa_j) \geq \frac{1}{\eps}\sum\limits_{j=1}^{v-1}
f\left(\kappa_ {\frac{j}{\eps}}\right)$. Thus the additional cost
is at most $\eps$ times the cost of $SOL_{inter}$ plus $f({\bf
\Delta})$.
\end{proof}

By concavity of $f$ we have $f(z)\leq z \cdot f(1)=z$ for any $z
\geq 1$. We have $|{\cal W'}|\leq |{\cal C}| \leq
\ell\cdot({1\over \eps^3}+1)^{1/\eps}$ and we also have $|H|\leq
|{\cal C}|$. If $S'=\emptyset$, we get
$k_{p_{\Delta}}=\frac{1}{\eps}$ so $f(k_{p_{\Delta}}) \leq
\frac{1}{\eps}$.


The cost of $SOL_{final}$ is at most $$(1+2\eps)\cdot \left(
(1+\eps^2)\cdot \left( (1+\eps)^2 \opt(I'')+f(k_{p_{\Delta}})\cdot
(|H|+2|{\cal W'}|) +2\eps^2 \opt(I'') \right) +f({1\over \eps})
\right)+f(k_{p_{\Delta}})$$ $$\leq (1+2\eps)(1+\eps)^5 \opt(I'') +
3f(k_{p_{\Delta}})\cdot (|H|+2|{\cal W'}|+1)+ 3f({1\over \eps}).$$
Therefore, the total cost of the returned solution (including the
cost of the packing of $S''$) is at most
\begin{eqnarray*} &&(1+2\eps)(1+\eps)^5 \opt(I'') + \left( {3(|H|+2|{\cal
W'}|+1)\over h(\eps)}+1\right) F(S'') + 3f({1\over \eps})\\
&\leq& (1+2\eps)(1+\eps)^5 \opt(I'')+ (1+3\eps)F(S'')+ 3f({1\over
\eps})\\
&\leq & (1+2\eps)(1+\eps)^5 (\opt(I'')+F(S''))+ 3f({1\over
\eps})\\ &=& (1+2\eps)(1+\eps)^5 \opt'(I')+ 3f({1\over \eps})\\
&\leq& (1+2\eps)(1+\eps)^6 \opt + (3h(\eps)+6)(1+2\eps)(1+\eps)^5
f({1\over \eps}) \\ &\leq& (1+2\eps)(1+\eps)^6 \opt + (18 \ell
({1\over \eps^3}+1)^{1/\eps} +6)(1+2\eps)(1+\eps)^5 {1\over
\eps^2}\\ &\leq &(1+2\eps)(1+\eps)^6 \opt + (18 ({1\over
\eps}+\log_{1+\eps} \opt)\cdot ({1\over \eps^3}+1)^{1/\eps}
+6)(1+2\eps)(1+\eps)^5 {1\over \eps^2}.\end{eqnarray*}  We note
that the last bound can be written as $(1+O(\eps))\cdot \opt +
t(\eps)\cdot \log_2 \opt + T(\eps)$ where $t$ and $T$ are some
(exponential) functions of $1\over \eps$.  To show that the
resulting scheme is an AFPTAS it suffices to argue that
$t(\eps)\cdot \log_2 \opt \leq \eps \opt + \left( {t(\eps)\over
\eps} \right) ^2$.  To see this last inequality note that if $
\log_2 \opt \leq {t(\eps) \over \eps^2}$ the claim clearly holds.
Otherwise, $\log_2 \opt \geq {t(\eps) \over \eps^2}$ and therefore
$\opt \geq 16$ (where the last inequality holds since
$t(\eps)>16$).  Note that for $x>16$ we have $\sqrt{x} \geq \log_2
x$ and by $\opt \geq 16$, we get $\sqrt{\opt} \geq \log_2 \opt
\geq {t(\eps) \over \eps^2}$. Therefore, $\eps \opt = \eps
\sqrt{\opt} \sqrt{\opt} \geq {t(\eps) \over \eps} \sqrt{\opt}\geq
{t(\eps) \over \eps} \log_2 \opt \geq t(\eps)\log_2 \opt $ and the
claim follows. Therefore, we have established the correctness of
Theorem \ref{afptasthm}.

\bibliographystyle{plain}


\end{document}